\documentclass[sigconf,10pt,printfolios=true,authorversion,nonacm]{acmart}
\renewcommand\footnotetextcopyrightpermission[1]{} % removes footnote with conference information in first column
\PassOptionsToPackage{english,capitalize,noabbrev}{cleveref}

\usepackage{hyperref}
\usepackage{style}
\crefname{line}{line}{lines}

\usepackage{color}

\newcommand{\strong}{strong PIP}
\newcommand{\weak}{relaxed PIP}

\newcommand{\dproperty}{Decomposable Property}

\newcommand{\naive}{\algname{RP-Na\"ive}}
\newcommand{\rpflat}{\algname{RP-Flat}}
\newcommand{\oneres}{\algname{RP-OneRes}}
\newcommand{\final}{\algname{RP-Final}}

\newcommand{\forkins}{\texttt{fork}}

\newcommand{\thread}{thread}

\newcommand{\pramresults}{~\cite{guan1991time,langston1993time,zheng1999efficient,guan1992parallel,huang1989stable,pietracaprina2015space}}

\newcommand{\depth}{span}

\newcommand{\Boruvka}{Bor\r{u}vka's}

\newcommand{\cmark}{\ding{51}}%
\begin{document}
\title{Parallel In-Place Algorithms: Theory and Practice}
  \author{Yan Gu}
  \affiliation{\institution{UC Riverside}}
  \email{ygu@cs.ucr.edu}

  \author{Omar Obeya}
  \affiliation{\institution{MIT CSAIL}}
  \email{obeya@mit.edu}
  \author{Julian Shun}
  \affiliation{\institution{MIT CSAIL}}
  \email{jshun@mit.edu}

\titlenote{ This is the full version of the paper appearing in the \emph{Proceedings of the SIAM Symposium on Algorithmic Principles of Computer Systems}, 2021.}

\begin{abstract}

Many parallel algorithms use at least linear auxiliary space in the
size of the input to enable computations to be done independently
without conflicts.  Unfortunately, this extra space can be prohibitive
for memory-limited machines, preventing large inputs from being
processed.  Therefore, it is desirable to design parallel in-place
algorithms that use sublinear (or even polylogarithmic) auxiliary space.

In this paper, we bridge the gap between theory and practice for
parallel in-place (PIP) algorithms.  We first define two computational
models based on fork-join parallelism, which reflect modern
parallel programming environments.  We then introduce a variety of new
parallel in-place algorithms that are simple and efficient, both
in theory and in practice.  Our algorithmic highlight is the
Decomposable Property introduced in this paper, which enables existing
non-in-place but highly-optimized parallel algorithms to be converted
into parallel in-place algorithms.  Using this property, we obtain
algorithms for random permutation, list contraction, tree contraction,
and merging that take linear work, $O(n^{1-\epsilon})$ auxiliary
space, and $O(n^\epsilon\cdot\text{polylog}(n))$ span for $0<\epsilon<1$.
We also present new parallel in-place algorithms for scan,
filter, merge, connectivity, biconnectivity, and minimum spanning forest using
other techniques.

In addition to theoretical results, we present experimental results
for implementations of many of our parallel in-place algorithms.  We
show that on a 72-core machine with two-way hyper-threading, the
parallel in-place algorithms usually outperform existing parallel
algorithms for the same problems that use linear auxiliary space,
indicating that the theory developed in this paper indeed leads to
practical benefits in terms of both space usage and running time.

\end{abstract}

\maketitle

\section{Introduction}

Due to the rise of multicore machines with tens to hundreds of cores
and terabytes of memory, and the availability of programming languages
and tools that simplify shared-memory parallel computing, many parallel
algorithms have been designed for large-scale data processing.
Compared to distributed or external-memory solutions, one of the
biggest challenge for using multicores for large-scale data processing
is the limited memory capacity of a single machine. Traditionally,
parallel algorithm design has mostly focused on solutions with low
work (number of operations) and \depth{} (depth or longest critical
path) complexities. However, to enable data to be processed in
parallel without conflicts, many existing parallel algorithms are not
in-place, in that they require $\Omega(n)$ auxiliary memory for an
input of size $n$.
For example, in the shuffling step of distribution sort (sample sort) or radix sort
algorithms, even if we know the destination of each element in the
final sorted array, it is difficult to directly move all of them to their
final locations in parallel in the same input array due to conflicts.
As a result, parallel algorithms for this task
(e.g.,~\cite{blelloch2010low,Blelloch91}) use an auxiliary array of
linear size to copy the elements into their correct final locations.

While many parallel multicore algorithms are work-efficient and have
low \depth{}, the $\Omega(n)$ auxiliary memory required by the algorithms
can prevent larger inputs from being processed. Purchasing or renting
machines multicore machines with larger memory capacities is an
option, but for large enough machines, the cost increases roughly
linearly with the memory capacity, as shown in
Figure~\ref{fig:price}. Furthermore, additional energy costs need to
be paid for machines that are owned, and the energy cost increases
proportionally with the memory capacity. Therefore, designing
\emph{parallel in-place (PIP) algorithms}, which use auxiliary space that is
sublinear (or even polylogarithmic) in the input size, can lead to
considerable savings. In addition, in-place algorithms can also reduce
the number of cache misses and page faults due to their lower memory
footprint, which in turn can improve overall performance, especially
in parallel algorithms where memory bandwidth and/or latency is a
scalability bottleneck.

\begin{figure*}[h]
\centering
  \includegraphics[width=.85\columnwidth]{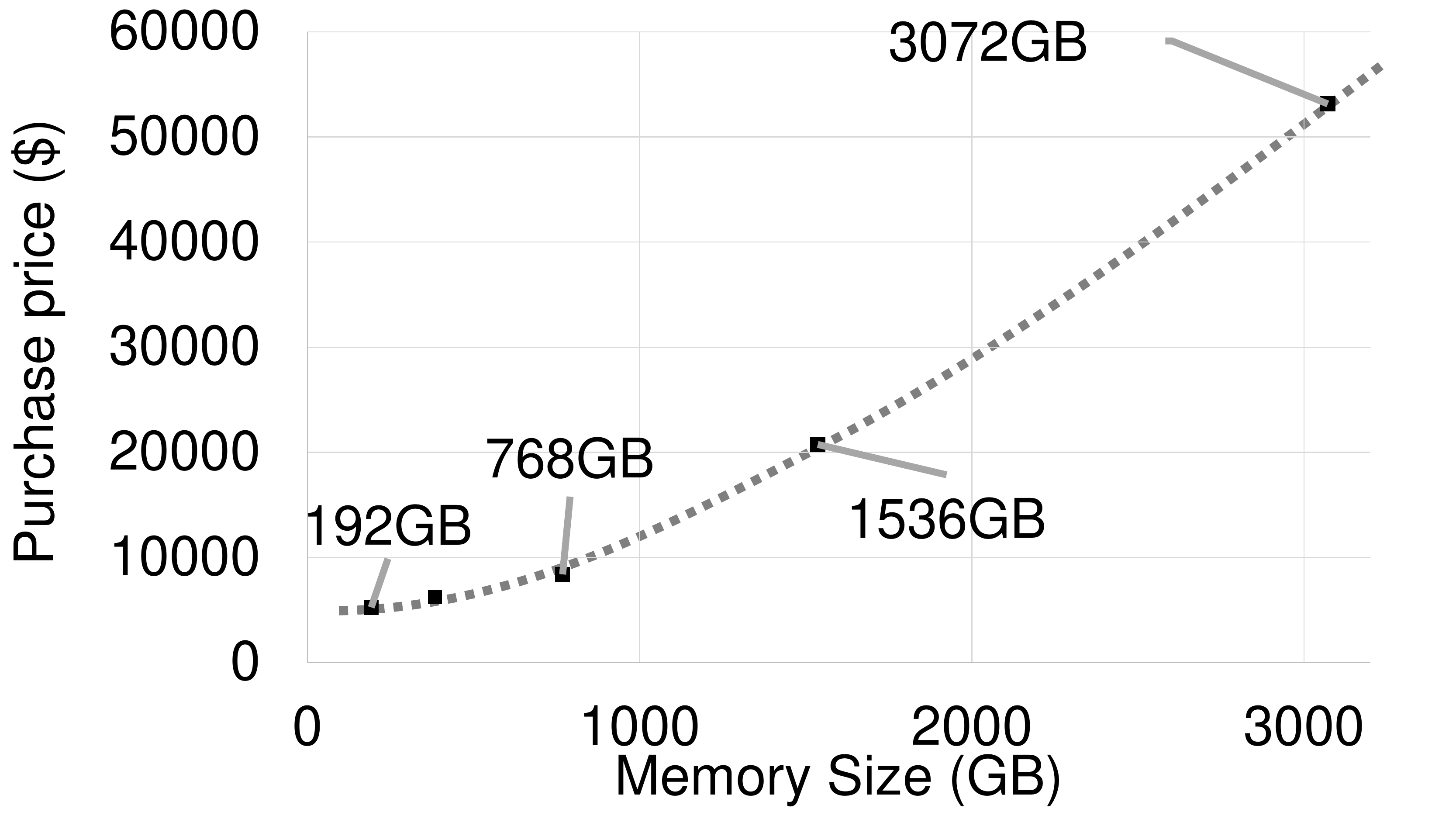}
  \includegraphics[width=.85\columnwidth]{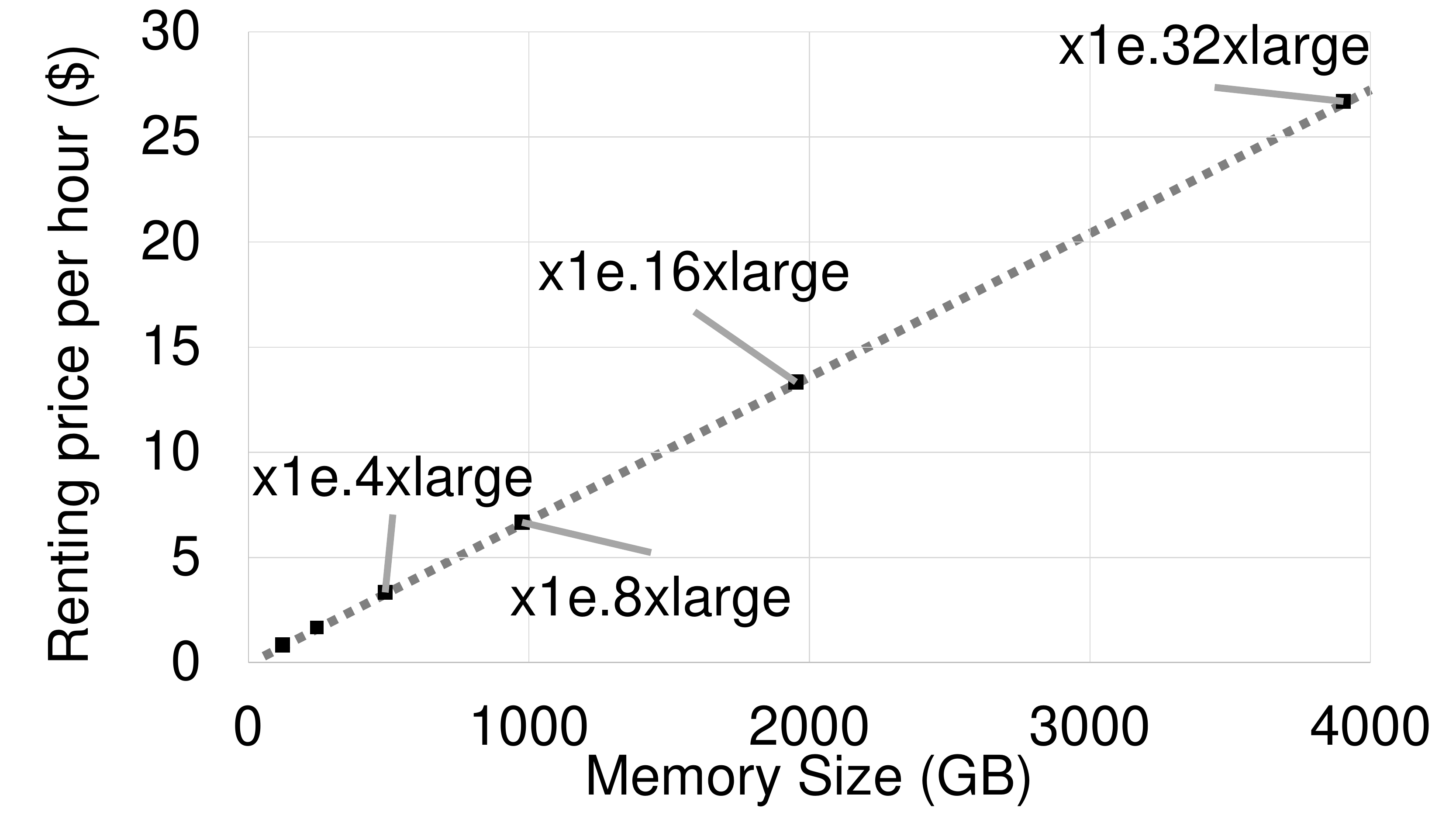}
%\vspace{-1em}
\caption{ \normalfont Purchase and rental prices for multicore
  servers as a function of memory size.
  The left figure shows the purchase price of an RAX
  XT24-42S1 server with 72 CPU cores (Xeon Gold 5220) for different
  DRAM sizes.  The right figure shows the rental price of AWS EC2
  x1e-series multicore instances vs. the DRAM capacity. In both
  cases, the DRAM capacity is the dominant part of the overall cost. }
\label{fig:price}
%\vspace{-1em}
\end{figure*}

There has been recent work studying theoretically-efficient and
practical parallel in-place algorithms for sample
sort~\cite{axtmann2017place}, radix
sort~\cite{obeya2019theoretically}, partition~\cite{kuszmaul2020cache},
and constructing implicit search tree layouts~\cite{berney2018beyond}.
These PIP algorithms achieve better performance than previous
algorithms in almost all cases.  While these algorithms are insightful
and motivate the PIP setting, they are algorithms designed for
specific problems and have different notions of what ``in-place''
means in the parallel setting.  In this paper, we generalize the ideas
in previous work into two models, which we refer to as the
\defn{strong PIP model} and the \defn{relaxed PIP model}.  At a high
level, the relaxed PIP model provides similar properties to the
classic in-place PRAM model, and the strong PIP model puts further
restrictions on memory allocation that allows PIP algorithms to
simultaneously achieve small auxiliary space and low span.  We provide
more details on these models in Section~\ref{sec:model}.

The main contribution of this paper is a collection of new PIP
algorithms in the two models, which include algorithms for solving
scan, merge, filter, partition, sorting, random permutation, list
contraction, tree contraction, and several graph problems
(connectivity, biconnectivity, and minimum spanning forest).  The
results are summarized in Table~\ref{tab:overview}, and discussed in
more detail in Sections~\ref{sec:decomp}--\ref{sec:relax}.  Some of
the algorithms are known, and we summarize them in this paper.  The
rest are new to the best of our knowledge, and we distinguish them by
presenting our results in theorems and corollaries.

\begin{table}[t]
  \small
  \centering
    \begin{tabular}{@{}@{}p{1.5cm}l@{}c@{}c@{}}
    \toprule
    \multicolumn{1}{c}{Model} & \multicolumn{1}{c}{Problems} & \multicolumn{1}{c}{Work-efficient} &  \\
    \midrule
    \multirow{6}[0]{1.5cm}{Strong PIP Model} & Permuting tree layout & \cmark & \cite{berney2018beyond} \\
          & Reduce, rotating & \cmark &  \\
          & Scan (prefix sum) & \cmark & * \\
          & Filter, partition, quicksort &       &  \\
          & Merging, mergesort &       &  \\
          & Set operations & \cmark & \cite{BlellochFS16} \\
          \midrule
    \multirow{7}[0]{1.5cm}{Relaxed PIP Model} & Random permutation & \cmark & * \\
          & List and tree contraction & \cmark & * \\
          & Merging, mergesort & \cmark & * \\
          & Filter, partition, quicksort & \cmark &  \\
       %   & Stable radix sort & \cmark &  \\
          & (Bi)Connectivity &       &  * \\
          & Minimum spanning forest &       & * \\
          \bottomrule
    \end{tabular}%%
    \caption{\label{tab:overview} Algorithms based on the \strong{} model and the \weak{} model.  ``Work-efficient'' indicates that the PIP algorithm has the same asymptotic work (number of operations) as the best sequential algorithm with no restriction on auxiliary space.
      Algorithms marked with * are our main contributions, while other algorithms are either known or  require slight changes to existing algorithms.  Merging and mergesort in the relaxed PIP model has been presented in~\cite{huang1988practical}, but our new algorithm in this paper is much simpler.  If a problem has a work-efficient solution in the \strong{} model, then it will not be listed again in the \weak{} model.}
    %\vspace{-2em}
%\julian{it is not clear which ones are new and which ones are not. what about the algorithms with no citation and no *}
    
\end{table}%

\begin{figure}[th]
%\vspace{-1em}
\centering
  \includegraphics[width=\columnwidth]{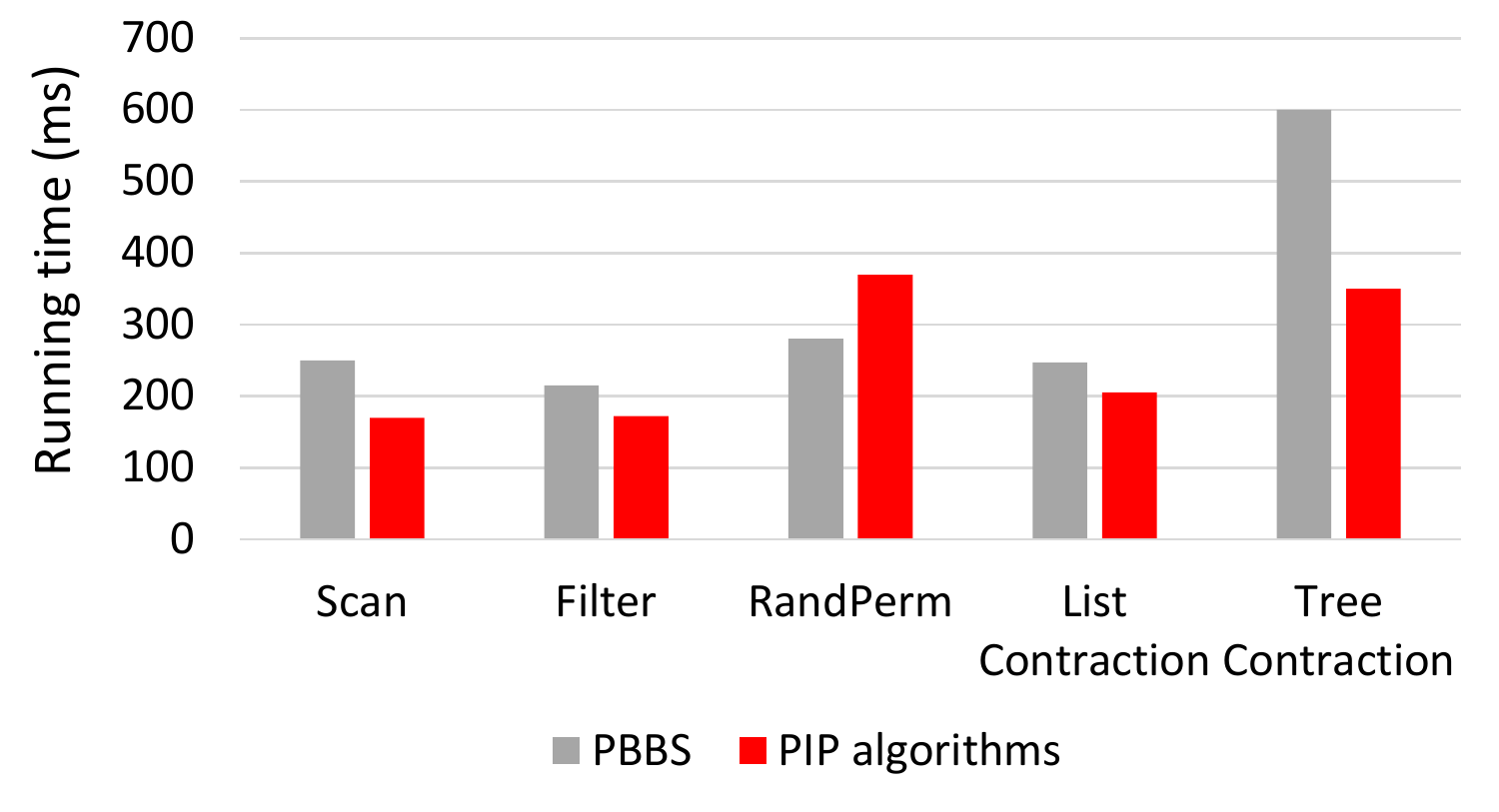}
\caption{\label{fig:summary} Running times for our new
  PIP algorithms compared to the non-in-place implementations from
  PBBS~\cite{shun2012brief}.
  %% Since most of the PIP algorithms in this
  %% paper are simple and practical, we implement some of them and
  %% compared them to the PBBS implementations, which are not in-place.
  For scan and filter, the input size is $10^9$, and for the other algorithms,
  the input size is $10^8$.  The running times are obtained on a 72-core machine
  with two-way hyper-threading, and more details are presented in
  Section~\ref{sec:exp}.  In all cases, the new PIP algorithms have
  competitive or better performance, with the additional advantage of
  using less auxiliary space.}
%\vspace{-1em}
\end{figure}

The algorithmic highlight in this paper is the \dproperty{} defined in
Section~\ref{sec:decomp}.  The high-level idea is to iteratively
reduce a problem to a subproblem of sufficiently smaller size, where the
the reduction can be performed using a non-in-place algorithm for the same problem.
If we can perform the reduction efficiently, then this leads to
 an efficient algorithm in
the \weak{} model.
This means that we can convert any existing non-in-place but highly-optimized
parallel algorithm to an efficient PIP algorithm.
We show many examples of this approach in this paper,
including algorithms for random permutation, list contraction, tree contraction, merging, and
mergesort.
We have also designed other PIP algorithms without
using the \dproperty{}, including algorithms for scan, filter, and various graph problems.

We implement five of our in-place algorithms, and compare
them to the optimized non-in-place implementations in the Problem
Based Benchmark Suite (PBBS)~\cite{shun2012brief}.  The running time
comparisons for certain input sizes are shown in
Figure~\ref{fig:summary} and we provide more details in
Section~\ref{sec:exp}.
We show that in addition to lower space usage, our in-place algorithms can have competitive or even
better performance compared to their non-in-place counterparts due to
their smaller memory footprint, indicating that the theory for PIP algorithms developed in this paper can lead to practical outcomes. Our implementations are publicly-available at {\small \url{https://github.com/ucrparlay/PIP-algorithms}}.

\section{Preliminaries}\label{sec:prelim}

\myparagraph{Work-Span Model.}  In this paper, we use the classic
work-\depth{} model for fork-join parallelism with binary forking for analyzing parallel
algorithms~\cite{CLRS,blelloch2020optimal}.
Unlike machine-based cost models such as the PRAM model~\cite{JaJa92}, this model is a language-based model, and we will justify the use of this model in \cref{sec:model}.
In this model, we assume that we have a set of \thread{}s that have
access to a shared memory.  Each \thread{} supports the same
operations as in the sequential RAM model, but also has a \forkins{}
instruction that forks two new child \thread{}s.  When a \thread{}
performs a \forkins{}, the two child \thread{}s all start by running
the next instruction, and the original \thread{} is suspended until
all of its children terminate.  A computation starts with a single {root}
\thread{} and finishes when that root \thread{} finishes.  The
\defn{work} of an algorithm is the total number of instructions and
the \defn{\depth{}} (depth) is the length of the longest sequence of dependent instructions in
the computation.
A thread can allocate a fixed size of memory that is either shared by all threads (referred to as ``heap-allocated'' memory), or private to this thread and other threads that it forks  (referred to as ``stack-allocated'' memory).
The latter case requires freeing the memory allocated after a fork and before the corresponding join.
A randomized work-stealing scheduler, which is widely used in real-world parallel languages such as Cilk, TBB, and X10, can execute an algorithm with work $W$ and \depth{} $D$ in $W/P+O(D)$ time \whp{}\footnote{We say $O(f(n))$ \defn{with high probability (\whp{})} to indicate $O(cf(n))$ with probability at least $1-n^{-c}$ for $c \geq 1$, where $n$ is the input size.} on $P$ processors~\cite{BL98,ABP01}.

In this paper, we also analyze the auxiliary space used, and we measure space in units of words.

\myparagraph{Work-Efficient and Low-Span Parallel Algorithms.}
The goal of designing a parallel algorithm is to achieve work-efficiency and low span.
Work-efficiency means that the algorithm asymptotically uses no more work than the best (sequential) algorithm for the problem.
Low span means that the longest sequence of dependent operations has polylogarithmic length.
Achieving low span can lead to practical benefits.
For instance, the number of steals in a work-stealing scheduler can be bounded by $O(PD)$, which is  proportional to the number of extra cache misses due to parallelism~\cite{Acar02,blelloch2010low}.
Low span also leads to fewer rounds of global synchronization, which can lead to significant performance improvements on modern architectures.

\myparagraph{The ``busy-leaves'' property.}
When using the Cilk work-stealing
scheduler, a fork-join program that uses $S_1$ words of space in a
stack-allocated fashion when run on one processor
will use $O(PS_1)$ words of space when run on $P$ processors~\cite{BL98}. This is a consequence of the ``busy-leaves'' property of the work-stealing scheduler.

\myparagraph{Problem definitions.}  Here we define the problems that are used
in multiple places in this paper.  Other problems are defined in their
respective sections.  Consider a sequence $[a_1, a_2, \ldots , a_n]$,
an associative binary operator $\oplus$, and an identity element $i$.
\defn{Reduce} returns $a_1 \oplus a_2 \oplus \ldots \oplus a_n$.
\defn{Scan} (short for an exclusive scan) returns $[i, a_1, (a_1
  \oplus a_2), \ldots , (a_1 \oplus a_2 \oplus \ldots \oplus
  a_{n-1})]$, in addition to the sum of all elements.  An inclusive
scan returns $[a_1, (a_1 \oplus a_2), \ldots , (a_1 \oplus a_2 \oplus
  \ldots \oplus a_n)]$.
\defn{Filter} takes an array $A$ and a
predicate function $f$, and returns a new array containing $a \in A$
for which $f(a)$ is true.
\defn{Partition} is similar to filter, but in addition to placing the
elements $a$ where $f(a)$ is true at the beginning of the array,
elements $a$ for which $f(a)$ is false are placed at the end of the
array.  We say that a filter or partition is \defn{stable} if the elements in the output are in the
same order as they appear in $A$.

\section{Models for Parallel In-Place Algorithms}\label{sec:model}

In the past, PIP algorithms have been designed based the
\emph{in-place PRAM model}~\pramresults.  However, this model and the
PRAM model itself have some limitations, which we describe in
Section~\ref{sec:previous-models}.  Hence, recent work on PIP
algorithms~\cite{axtmann2017place,obeya2019theoretically,kuszmaul2020cache,berney2018beyond}
incorporate the PIP setting in the newer work-span model, although
they use different notions of ``in-place'' for the algorithms.  In
this paper, we generalize the ideas into two models, which we refer to
as the \defn{strong PIP model} and the \defn{relaxed PIP model}.  At a
high level, the relaxed PIP model provides similar properties as the
in-place PRAM model, and the strong PIP model puts further
restrictions on memory allocation that enables PIP algorithms to
achieve small auxiliary space and low span simultaneously.  Based on
our model definitions, algorithms
in~\cite{kuszmaul2020cache,berney2018beyond} can be mapped to the
strong PIP model, and algorithms
in~\cite{axtmann2017place,obeya2019theoretically} can be mapped to the
relaxed PIP model.  In this section, we will first define these two
models, and then discuss their relationship to existing PIP models.

\subsection{The Strong PIP Model}

We start by defining the
\strong{} model based on the work-span model for fork-join parallelism.

\begin{definition}[Strong PIP model and algorithms]
The \textbf{\strong{} model} assumes a fork-join computation only using
$O(\log n)$-word auxiliary space in a stack-allocated
fashion for an input size of $n$ when run sequentially (with no auxiliary heap-allocated space).  We say that an algorithm is
\textbf{\strong{}} if it runs in the \strong{} model and has
polylogarithmic \depth{}.
\end{definition}

For a PIP algorithm in the \strong{} model, the Cilk work-stealing
scheduler can bound the total auxiliary space to be
$O(P\log n)$ words, where $P$ is the number of processors~\cite{BL98}.
All \strong{} algorithms presented in this paper, as well as existing ones~\cite{berney2018beyond,kuszmaul2020cache}, only use $O(\log n)$-word stack-allocated auxiliary space sequentially.
We say a \strong{} algorithm is \defn{optimal} if its work and \depth{} bounds match the best non-in-place counterpart.

\subsection{The Relaxed PIP Model}

Many existing PIP algorithms~\cite{guan1991time,langston1993time,zheng1999efficient,guan1992parallel,huang1989stable,pietracaprina2015space,axtmann2017place,obeya2019theoretically} exhibit a tradeoff between additional space $S$ and \depth{} $D$, such that $S\cdot D=\tilde{\Theta}(n)$.\footnote{We use $\tilde{O}(f(n))$ to hide polylogarithmic factors.}
We capture these algorithms in our \weak{} model, and refer to these algorithms as \weak{} algorithms.

\begin{definition}[Relaxed PIP model and algorithms]
The \textbf{\weak{} model} assumes a fork-join computation using
$O(\log n)$-word stack-allocated space sequentially and $O(n^{1-\epsilon})$ shared (heap-allocated) auxiliary space for an input of size~$n$ and some constant $0<\epsilon<1$.
We say that an algorithm is \textbf{\weak{}} if it runs in the
\weak{} model and has $O(n^{\epsilon}\cdot \polylog(n))$ \depth{} for all values of $\epsilon$.
\end{definition}

The Cilk work-stealing scheduler can bound the total auxiliary space of \weak{} algorithms to be
$O(n^{1-\epsilon}+P\log n)$ on $P$ processors.
For brevity, we refer to the auxiliary space in future references to the \weak{} model as just the heap-allocated space.
Algorithms in the \weak{} model allow sublinear auxiliary space, which is less
restrictive than in the \strong{} model. This provides
more flexibility in algorithm design, while still being
useful in practice as \weak{} algorithms still use less space than
their non-in-place counterparts.  In the next section, we introduce a
general property, which allows any existing parallel algorithm with
polylogarithmic span that satisfies the property to be easily
converted into a \weak{} algorithm.

\subsection{Relationship to Previous Models}\label{sec:previous-models}

PIP algorithms have been analyzed in the in-place PRAM model for decades. Recent work~\cite{axtmann2017place,obeya2019theoretically,kuszmaul2020cache,berney2018beyond} has designed in-place algorithms into the work-span model, but they only provide algorithms for specific problems rather than focusing on the general
parallel in-place setting. 
In this paper, we formally define the parallel in-place models, and justify the models
by discussing the limitations of the previous in-place PRAM, and how our new models overcome it.

\myparagraph{The in-place PRAM.}
Most existing parallel in-place algorithms have been designed in the in-place PRAM~\cite{guan1991time,langston1993time,zheng1999efficient,guan1992parallel,huang1989stable,pietracaprina2015space}.
The PRAM has $P$ processors that are fully synchronized between steps, and
the running time of an algorithm is the maximum number of steps $T$
used by any processor.  In this model, the auxiliary space~$S$ is the sum of the total space
used across all processors.
As pointed out by Berney et al.~\cite{berney2018beyond}, each processor on a PRAM requires $\Omega(1)$ (usually $\Omega(\log n)$) auxiliary space to do anything useful (e.g., storing the program counter and using registers).
This indicates that if the total auxiliary space $S$ for all processors is bounded to be small, then the parallelism is also bounded by $O(S)$.
This is because even if we have an infinite number of processors, no more than $S$ of them can do useful work simultaneously.
The overall PRAM time is $\Omega(W/S)$, where $W$ is the overall work in the algorithm.
Hence, in the PRAM setting, an algorithm can only achieve high parallelism when $S$ is \emph{asymptotically close} to $W$.
This has been described by Langston et al.~\cite{langston1993time,zheng1999efficient,guan1992parallel,huang1989stable} as the time-space tradeoff in the PRAM---if the input size is $n$, then the product of auxiliary space $S$ and PRAM time $T$ is $\tilde{\Omega}(n)$, and an algorithm is optimal on a PRAM when $S\cdot T=\tilde{\Theta}(n)$.
This limitation arises because the analysis of parallelism and auxiliary space are intertwined in the in-place PRAM.

\myparagraph{Decoupling the analysis between parallelism and auxiliary space.}
Parallel algorithms with low span have many practical benefits even
for small processor counts, due to lower scheduling overhead and improved cache locality, as discussed in \cref{sec:prelim}.  However, low span cannot
be achieved in the in-place PRAM unless we use nearly linear auxiliary
space.  Our goal is to decouple the analysis of parallelism from the
restriction of auxiliary space.  In both the \strong{} and \weak{} models, the
auxiliary space in measured in the sequential setting, whereas the span is analyzed
based on the fork-join computation graph. This decouples the space analysis from the span analysis.
Furthermore, in the \strong{} model, low span and
small auxiliary space can be achieved simultaneously.

To achieve the decoupling, we use  the separation of the private ``stack-allocated'' memory from the shared ``heap-allocated'' memory in work-span model.
The heap-allocated memory is what we usually refer to as the shared memory, and is independent of the number of processors.
The stack-allocated memory is per processor, and the ``busy-leaves'' property guarantees that the overall space usage of a program is $O(PS_1)$ when it is run on  $P$  processors, where $S_1$ is the amount of stack-allocated memory when running the algorithm sequentially.
Since $P$ is usually modest in practice, if the stack-allocated memory is small (e.g., $O(\log n)$), then the auxiliary space size $O(PS_1)$ will be negligible on modern machines.
As a result, the abstraction of the stack-allocated memory separates the per-processor need from the shared resource, and overcomes the limitation of the in-place PRAM by dynamically mapping the algorithm on a machine with $P$ processors, with the auxiliary space guarantee.

In addition to the advantages discussed above,
the work-span model simplifies parallel algorithm design and analysis, and algorithm designers do not need to worry about low-level details related to hardware such as memory allocation, caching, and load balancing.
Recent papers~\cite{berney2018beyond,kuszmaul2020cache,obeya2019theoretically} have made a similar observation on the limitation of the in-place PRAM model, and analyzed the in-place algorithms using the work-span model.
In this paper, we explicitly formalize this discussion and define the two PIP models based on the work-span model.

\myparagraph{Other practical considerations.}
Here we describe additional benefits to use the new PIP models based on the work-span model.
Modern parallel programming languages, such as Cilk, OpenMP, TBB, and X10, directly support algorithms designed for the work-span model using fork-join parallelism, with efficient runtime schedulers.
In contrast to the PRAM, in which computations have many synchronization points,
computations in the work-span model can be highly asynchronous. This is a practical advantage due to the high synchronization overheads on modern hardware~\cite{blelloch2020optimal}.
Furthermore, the PIP algorithms in this paper based on our new models have additional guarantees with respect to multiprogrammed
environments~\cite{ABP01}, cache complexity~\cite{Acar02,BlellochFinemanGibbonsEtAl2011,blelloch2010low},
write-efficiency~\cite{BBFGGMS16,blelloch2015sorting,blelloch2016efficient}, and resource-obliviousness~\cite{Cole2017}.

\section{\dproperty{}}\label{sec:decomp}

Designing \strong{} algorithms is generally challenging (we present several in Section~\ref{sec:strong}, but if we relax the auxiliary space to sublinear (the \weak{} model), then we believe that PIP algorithms can be designed for many more problems.
In this section, we introduce the \defn{\dproperty{}}, which enables  any existing parallel algorithm that satisfies the property to be converted into a \weak{} algorithm.
If the existing parallel algorithm is work-efficient, then the corresponding \weak{} algorithm will also be work-efficient.

\begin{theorem}[\dproperty{}]\label{thm:generating}
  Consider a problem with input size $n$ and a parallel algorithm to solve it with work $W(n)=O(n\cdot \polylog(n))$. Let $r=W(n)/n$.
  If the problem can be reduced to a subproblem of size $n-n^{1-\epsilon}/r$ using $n^{1-\epsilon}$ work and space for some $0<\epsilon<1$, and polylogarithmic \depth{} $D(n)$, then there is a \weak{} algorithm for this problem with $W(n)$ work, $O(n^\epsilon\cdot\polylog(n))$ span, and $O(n^{1-\epsilon})$ auxiliary space.
\end{theorem}

\begin{proof}
  We iteratively reduce the problem size by $n^{1-\epsilon}/r$
  (this size remains the same throughout the algorithm), and each round takes $O(n^{1-\epsilon})$ work and space.
  Since $r=\polylog(n)$, this means $n^{1-\epsilon}$ is asymptotically larger than $r$, and
  we can reduce the problem size by at least one on each round.
  By applying this reduction for $rn^{\epsilon}$ rounds, we have a \weak{} algorithm with $W(n)$ work and $D(n)\cdot rn^\epsilon=O(n^\epsilon\cdot \polylog(n))$ \depth{}, using $O(n^{1-\epsilon})$ auxiliary space.
\end{proof}

The high-level idea of the \dproperty{} is that, for a problem of size $n$, if we can reduce the problem size to $n-n'$ using work proportional to $n'$, then we can control the additional space by varying the size of $n'$ to fit in the auxiliary space.
This provides theoretically-efficient \weak{} algorithms for parallel algorithms that satisfy this property.
On the practical side, we observe that this reduction step usually corresponds to solving a subproblem that is the same as the original problem but with a smaller size.
Hence, we can use the best existing non-in-place algorithms for this step. We show in Section~\ref{sec:exp} that the performance of our \weak{} algorithms using this approach is competitive or faster than their non-in-place counterparts.
In the rest of this section, we introduce some algorithms that satisfy the Decomposable Property.

\subsection{Random Permutation}

Generating random permutations in parallel is a useful subroutine in many parallel algorithms.
Many parallel algorithms (e.g., randomized incremental algorithms) require randomly permuting the input elements to achieve strong theoretical guarantees.
The sequential Knuth~\cite{Knuth69,Durstenfeld1964} shuffle algorithm, shown below, has linear work, where $H[i]$ is an integer uniformly drawn from $[1,\ldots, i]$, and $A$ is the array to be permuted.

{%\footnotesize
\setlength{\interspacetitleruled}{0pt}%
\setlength{\algotitleheightrule}{0pt}%
\begin{algorithm2e}[h]
\DontPrintSemicolon
\Fn {\mf{Knuth-Shuffle}(A, H)} {
\lFor {$i \leftarrow n$ to $1$} {$A[i]\gets i$}
\lFor {$i \leftarrow n$ to $1$} {
    swap($A[i]$, $A[H[i]]$)
}
}
\end{algorithm2e}
}

\begin{algorithm2e}[!t]
  %\small
\caption{$\mf{Parallel Knuth-Shuffle}(A, H)$~\cite{shun2015sequential}}
\label{alg:parallel-rp}
\SetKwFor{ParForEach}{parallel foreach}{do}{endfch}
\SetKwFor{ParFor}{parallel for}{do}{endfch}
\SetKw{Break}{break}
  \DontPrintSemicolon
$R\gets\{-1,\ldots,-1\}$\\
\lParFor {$i \leftarrow n$ to $1$} {
    $A[i]\gets i$
}
\While {swaps unfinished} {
  \ParForEach {unfinished swap $(s, H[s])$\label{line:resstart}} {
    $R[s]\gets \max(R[s],s)$\label{line:line6}\\
    $R[H[s]]\gets \max(R[H[s]],s)$\label{line:resend}
  }
  \ParForEach {unfinished swap $(s, H[s])$} {
       \lIf {$R[s] = s$ and $R[H[s]] = s$\label{line:ressuc}} {
           swap$(A[H[s]], A[s])$
       }
    }
  Reset $R$ and pack the leftover swaps (without modifying the swaps)

}
\Return {A}
\end{algorithm2e}

Recent work~\cite{shun2015sequential} has shown that this sequential iterative algorithm is readily parallel.
The pseudocode of this parallel algorithm is shown in Algorithm~\ref{alg:parallel-rp}, and is both theoretically and practically efficient.
The key idea is to allow multiple swaps to be performed in parallel as long as the sets of source and destination locations of the swaps are disjoint.
We illustrate the dependence structure on an example in Figure~\ref{fig:rand}.  Given an input array $H$, in Figure~\ref{fig:rand}(a) we create a node for each index, and an edge from the node to the node corresponding to its swap destination.
In this example, we can swap locations 6 and 3, 8 and 2, and 7 and 4 simultaneously in the first step since these three swaps do not interfere with each other.
To resolve the case where multiple nodes point to the same swap destination, we chain these nodes together, as shown
in Figure~\ref{fig:rand}(b). We also remove self-loops.
In Algorithm~\ref{alg:parallel-rp}, each unfinished swap writes to an auxiliary array $R$  using a $\max()$ to reserve both its source and destination locations (Lines~\ref{line:resstart}--\ref{line:resend}).
We assume that $\max()$ takes $O(1)$ work, and in practice, it can be implemented using a compare-and-swap loop~\cite{shun13reducing}.
We then perform the actual swaps in parallel for the swaps that successfully reserve both of its locations (Line~\ref{line:ressuc}).
The rest of the swaps will be packed and will try again in the next step.
Shun et al.\ show that Algorithm~\ref{alg:parallel-rp} finishes in $O(\log n)$ rounds \whp{}~\cite{shun2015sequential}.
The work and \depth{} can be shown to be $O(n)$ in expectation and $O(\log n)$~\whp{}, respectively~\cite{shun2015sequential,blelloch2020optimal}.

\begin{figure}[t]
%\vspace{-.5em}
\begin{center}
  \includegraphics[width=\columnwidth]{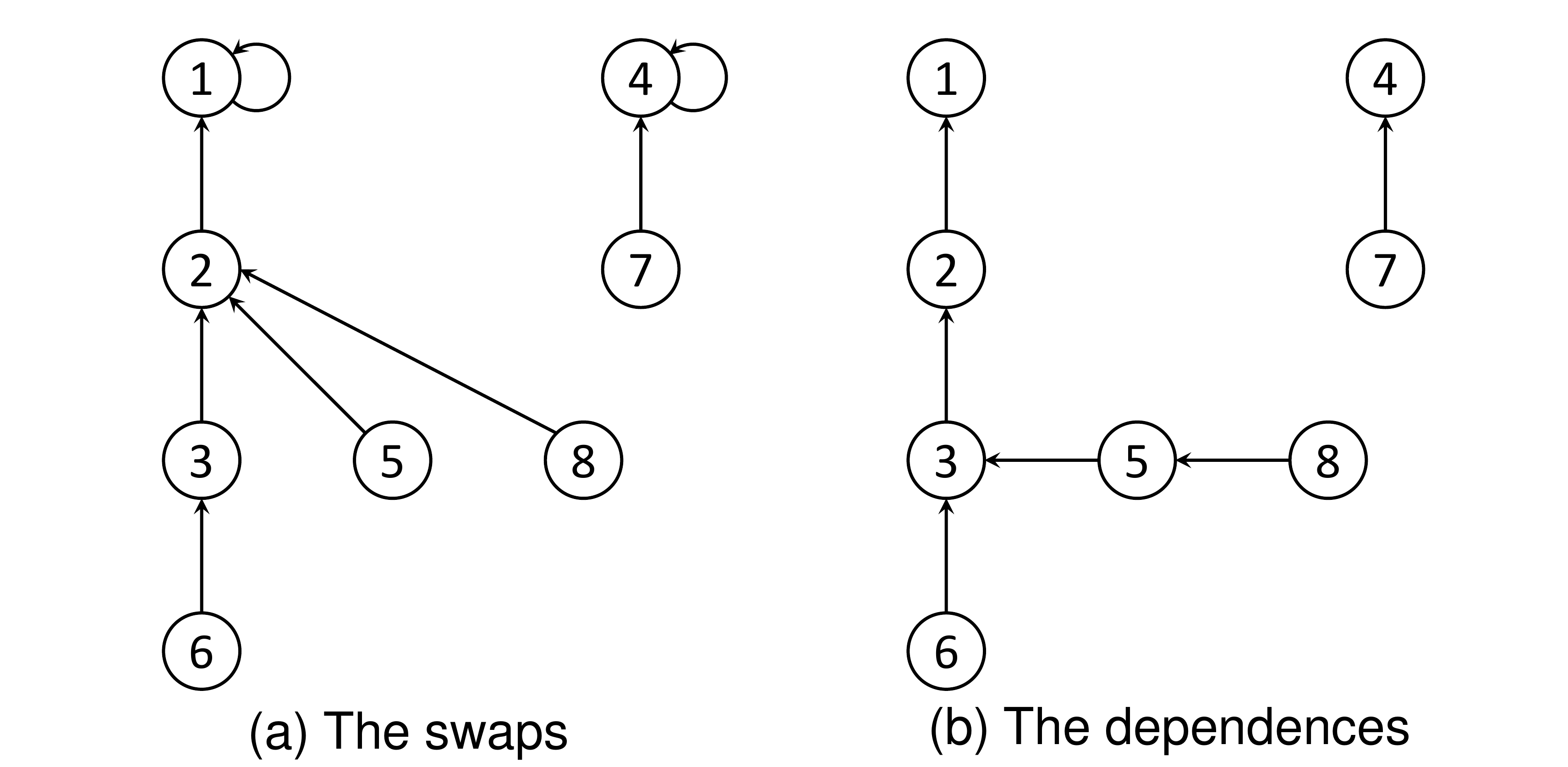}
\end{center}%\vspace{-.5em}
\caption{An example for $H = [1, 1, 2, 4, 2, 3, 4, 2]$.  (a) indicates the destinations of the swaps according to $H$.  The dependences of the swaps are shown in (b), indicating the order of the swaps.  %Figure~(c) links the roots of the forest to make it as a binary tree.
}
%\vspace{-1em}
\label{fig:rand}
\end{figure}

We now show that  the random permutation algorithm above satisfies  the \dproperty{}.
The property for the sequential Knuth shuffle is easy to see---after applying the first $n^{1-\epsilon}$ swaps, which we refer to as one \defn{round}, the problem reduces to a subproblem of size $n-n^{1-\epsilon}$, which can be solved using the same algorithm.
We note that for any $n^{1-\epsilon}$ swaps, up to $3n^{1-\epsilon}$ locations will be accessed in the $R$ and $H$ arrays ($H[s]$, $R[s]$, and $R[H[s]]$ for each swap source $s$).
We can use a parallel hash table to store these values using $O(n^{1-\epsilon})$ space.
When the load factor of the hash table is no more than one-half, then each update or query requires $O(1)$ expected work and $O(\log n)$ work~\whp{}~\cite{knuth1963notes,shun2014phase}.
To guarantee that our algorithm has the same bounds as proved in~\cite{shun2015sequential,blelloch2020optimal}, we always work on the first  $n^{1-\epsilon}$ unfinished swaps based on the sequential order.
The longest dependence length among the first $n^{1-\epsilon}$ swaps in a phase is bounded by $O(\log n)$ \whp{} since it cannot be longer than the overall dependence length for all $n$ swaps, which is bounded by $O(\log n)$ \whp{}.
The overall \depth{} in a phase is $O(\log^2 n)$, where the additional factor of $\log n$ due to hash table insertions and queries.
The entire algorithm finishes after $n^\epsilon$ rounds and is work-efficient.
By applying Theorem~\ref{thm:generating}, we obtain the following theorem.

\begin{theorem}
  There is a \weak{} algorithm for random permutation using $O(n)$ expected work, $O(n^\epsilon\log^2 n)$ \depth{} \whp{}, and $O(n^{1-\epsilon})$ auxiliary space  for $0<\epsilon<1$.
\end{theorem}

\myparagraph{Constant-dimension linear programming and smallest
  enclosing disks.}  Based on the \weak{} algorithm for random
permutation, it is straightforward to design \weak{} algorithms for
constant-dimension linear programming and smallest enclosing disks
using randomized incremental
construction~\cite{seidel1993backwards,blelloch2016parallelism}.  The
randomized algorithms after randomly permuting the input elements
take $O(d! n)$ expected work and $O(d\log n)$ span and auxiliary space \whp{}, where $d$ is the dimension~\cite{blelloch2016parallelism}, by using the in-place reduce algorithm that will be discussed in \cref{sec:strong}.
By using the \weak{} random permutation algorithm, we can obtain parallel in-place
algorithms for constant-dimension linear programming and smallest
enclosing disks in $O(d!n)$ expected work and  $O(n^\epsilon\log^2 n+d\log n)$ \depth{} \whp{}, using $O(n^{1-\epsilon}+d\log n)$
auxiliary space.

\subsection{List Contraction and Tree Contraction}\label{sec:list}

List ranking~\cite{Reid-Miller93,KarpR90,JaJa92}
is one of the most important problems in the study of parallel
algorithms.  The problem takes as input a set of linked lists, and
returns for each element its position in its list.
List contraction is used to contract a linked list into a single node, and is used as a subroutine in list ranking.

We now discuss the \dproperty{} of list contraction.  The order of
contracting elements does not matter as long as all elements are
eventually contracted.  Therefore, similar to random permutation, we
can process $n^{1-\epsilon}$ elements in a round, and apply
existing parallel list contraction
algorithms~\cite{KarpR90,JaJa92}
to contract these $n^{1-\epsilon}$ elements.
To show an example, we discuss Shun et al.'s non-in-place list contraction algorithm~\cite{shun2015sequential} and how to turn it into a \weak{} algorithm.
This is also the algorithm that we implemented in this paper. 

\begin{figure*}[t]
\vspace{-1em}
\begin{center}
  \includegraphics[width=2\columnwidth]{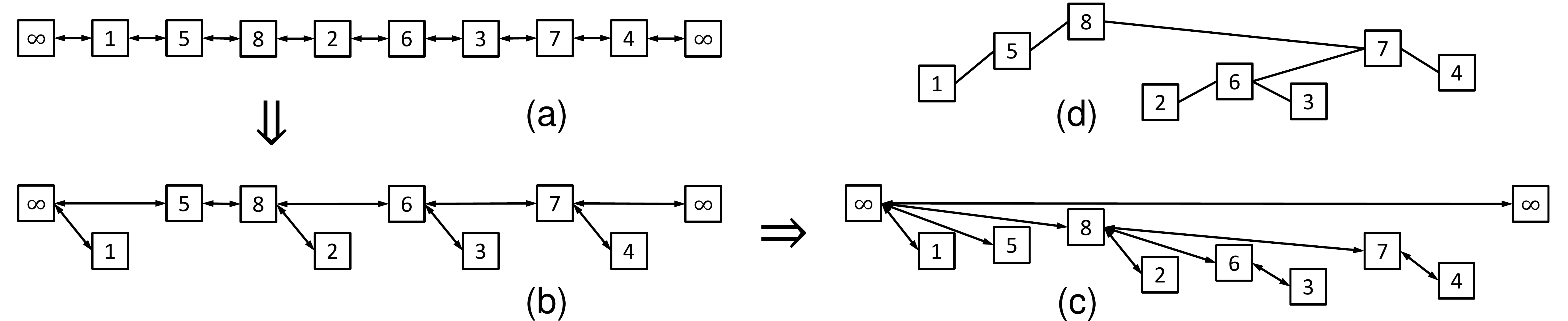}
\end{center}%\vspace{-1.2em}
\caption{A list contraction example for \cref{algo:lc} with the
  priorities shown in the boxes.  (a) shows the original input, with an extra $\infty$ element at each end of the list.  The
  first round of the algorithm contracts the nodes with priorities 1, 2,
  3, and 4, as shown in (b).  After all nodes are contracted, a tree
  structure is formed, as shown in (c).  The dependences of the
  algorithm are shown in (d), where a node depends on all of its descendants
  in the tree.  For instance, the contraction of the node with priority 6 needs to be
  done after the contraction of the nodes with priorities 2 and 3.
  %\julian{swap position of figures (c) and (d)}
}
%\vspace{-1em}
\label{fig:list}
\end{figure*}

\begin{algorithm2e}[t]
  %\small
\caption{$\mf{Parallel List-Contraction}(L)$~\cite{shun2015sequential}}
\label{algo:lc}
\SetKwFor{ParForEach}{parallel foreach}{do}{endfch}
\SetKwFor{ParFor}{parallel for}{do}{endfch}
\SetKw{Break}{break}
\KwIn{A doubly-linked list $L$ of size $n$.  Each element $l_i$ has a random priority $p(l_i)$.}
  \DontPrintSemicolon
%    \vspace{0.5em}
    $R\gets \{0,\ldots,0\}$\\
\While {elements remaining} {
  \ParForEach {uncontracted element $l_i$} {
    \If {\upshape $p(l_i)<p(\mb{prev}(l_i))$ and $p(l_i)<p(\mb{next}(l_i))$} {
        $R[i]\gets 1$
    }
  }
  \ParForEach {uncontracted element $l_i$} {
       \If {$R[i] = 1$} {
           Splice out element $l_i$ and update pointers
       }
    }
  Pack the leftover (uncontracted) elements
}
\Return {A}\end{algorithm2e}

The pseudocode and the high-level idea of this algorithm is given in Algorithm~\ref{algo:lc}.
A careful implementation of this algorithm takes
worst-case linear work and $O(\log n)$ span~\whp{}~\cite{shun2015sequential,blelloch2020optimal}.
This algorithm assigns a random priority to each list element (\cref{fig:list}(a)), contracts all elements that have priority lower than both of its neighbors' priorities (\cref{fig:list}(b)), packs the leftover elements, and iterates until the list is empty.
The number of rounds of this algorithm is the length of the longest dependence among the nodes, which is $O(\log n)$ \whp{}~\cite{shun2015sequential}.
\cref{fig:list}(d) shows the dependences in the example (a node depends on all of its descendants in the tree shown)---here the algorithm finishes in 4 rounds (the height of the tree).

As discussed, the order of the contraction does not matter.
Hence, for a problem of size $n$, we can work on $n^{1-\epsilon}$ elements
and contract them using this algorithm, which requires
$O(n^{1-\epsilon})$ work, $O(\log n)$
span~\whp{} (no more than the span for $n$ elements),  and $O(n^{1-\epsilon})$  space.  Then the
problem reduces to a subproblem of size $n-n^{1-\epsilon}$.  We can
iteratively apply this for $n^\epsilon$ rounds, which yields a
\weak{} algorithm for list contraction.

After the $n$ elements are spliced out, list contraction algorithm generates a tree, and the tree for the example in \cref{fig:list}(a) is shown in \cref{fig:list}(c). 
The remaining
work in list ranking after list contraction is referred to as
``reconstruction''~\cite{JaJa92}, which distributes the values down the
tree. 
Therefore, once we obtain this tree structure, the classic algorithms~\cite{JaJa92,Reif1993} for reconstruction take worst-case linear work and $O(\log n)$ span \whp{}.
Representing the
tree only requires $n$ pointers, which fit into the $2n$ pointers in
the input linked list if we are allowed to overwrite the input.
For our new \weak{} algorithm, we can store
the $n^{1-\epsilon}$ tree pointers in each round by overwriting the $2n^{1-\epsilon}$ pointers of the elements being processed in the current round.
After we recursively solve the smaller subproblem, we
can use the classic reconstruction algorithm for the $n^{1-\epsilon}$
elements in the current round, which takes worst-case linear  work and logarithmic span \whp{}.
In total, the reconstruction step has the same work, span, and auxiliary space bounds as list contraction.

Tree contraction is  a generalization of list contraction and has many applications in parallel tree and graph
algorithms~\cite{Reid-Miller93,JaJa92,MillerReif1985,shun2015sequential}.
Here we will assume that we are contracting rooted binary trees in which every internal node has exactly two children.
As in list contraction, the ordering of contracted tree nodes does not matter as long as a parent-child pair is not contracted in the same round.
For a problem of size $n$, we can work on $n^{1-\epsilon}$ tree nodes  each round and contract them using existing tree contraction algorithms, and repeat for $n^\epsilon$ rounds.
Therefore, the \dproperty{} is satisfied for tree contraction.
We can convert the parallel tree contraction algorithm of Shun et al.~\cite{shun2015sequential,blelloch2020optimal} that is not in-place, but theoretically and practically efficient, to a \weak{} algorithm that requires $O(n^{1-\epsilon})$ expected work and  $O(\log n)$ \depth{} \whp{}  per round, and $O(n^{1-\epsilon})$ space.

We obtain the following theorem for list contraction and tree contraction.

\begin{theorem}
  There are \weak{} algorithms for list contraction and tree contraction that take $O(n)$ work, $O(n^\epsilon\log n)$ span \whp{}, and $O(n^{1-\epsilon})$ auxiliary space  for $0<\epsilon<1$.
\end{theorem}

\subsection{Merging and Mergesort}\label{sec:merge}

Merging two sorted arrays of size $n$ and $m$ (stored consecutively in an array of size $n+m$) is another canonical primitive in parallel algorithm design. We assume without loss of generality that $n\ge m$.
Parallel in-place merging algorithms have been studied for the PRAM model~\cite{guan1991time}, using $O(n\log n)$ work, $O(n^\epsilon\log n)$ span, and $O(n^{1-\epsilon})$ auxiliary space when mapped to the \weak{} model.
However, this algorithm is quite complicated and unlikely to be practical.
By using the \dproperty{}, we can design a much simpler algorithm based on any existing textbook parallel non-in-place merging algorithm, combined with some features of the sequential in-place merging algorithm~\cite{huang1988practical}.
The key idea in~\cite{huang1988practical} for in-place merging  is to split both input arrays into chunks of size $k$, and sort the chunks based on the last element of each chunk.
Then, the algorithm merges the first remaining chunk from each of the two input arrays, and when one chunk is used up, the algorithm replaces it with the next chunk in the corresponding array.

To obtain a \weak{} algorithm, we set the chunk size to $k=n^{1-\epsilon}$, so that we can process two chunks using $O(n^{1-\epsilon})$ auxiliary space.
With this space bound, we can use a non-in-place merging algorithm to output the smallest $k=n^{1-\epsilon}$ elements and repeat for $O(n^\epsilon)$ rounds.

The first step of our algorithm is the same as~\cite{huang1988practical}, which
 sorts the chunks based on only their last elements, and moves each chunk to their final destination in parallel by using the $O(n^{1-\epsilon})$ auxiliary space as a buffer. Sorting all of the chunks takes $O(n^\epsilon\log n)$ span.
Then, in the merging phase, we move the first chunk from each array to the auxiliary space, use any existing parallel merging algorithm to merge them, until we either run out of the elements in one chunk, at which point we load the next chunk of the corresponding array to the auxiliary space, or until we gather a full chunk of merged elements, at which point we flush it back to the original array and empty the buffer.
At any time, there can be at most three chunks in the auxiliary space---two chunks from the input arrays and one chunk for the merged output, and so the required auxiliary space is $O(n^{1-\epsilon})$.
We can use any existing non-in-place parallel algorithm~\cite{blelloch2018introduction,JaJa92} to perform the merge in the auxiliary space, which takes linear work  and logarithmic \depth{}. 
Such calls to merge in the algorithm can happen at most $2(n+m)/k$ times---$(n+m)/k$ times after loading new chunks to the auxiliary space and $(n+m)/k$ times after the output chunk is full and is flushed.
Each merge takes work linear in the output size, and $O(\log k)=O(\log n)$ span.
The overall work is therefore $O(n)$, and the \depth{} is $O((n/k) \log n)=O(n^\epsilon\log n)$. This gives the following theorem.

\begin{theorem}
 Merging  two sorted arrays of size $n$ and $m$ (where $n\ge m$) stored consecutively in memory takes $O(n)$
 work, $O(n^\epsilon\log n)$ span, and $O(n^{1-\epsilon})$ auxiliary
 space for $0<\epsilon<1$.
\end{theorem}

When $\epsilon>1/2$, the auxiliary space $O(n^{1-\epsilon})$ is insufficient for sorting all $O(n^\epsilon)$ chunks at the beginning, and so we sort $O(n^{1-\epsilon})$ chunks at a time until all chunks have been processed. As done in~\cite{huang1988practical}, we use dual binary search to find the smallest  $O(n^{1-\epsilon})$ chunks to merge, and repeat our above algorithm for $O(n^\epsilon)$ rounds.
This will not affect the cost bounds.

With the \weak{} merging algorithm, we can obtain a \weak{} mergesort algorithm with $O(n\log n)$ work, $O(n^{1-\epsilon})$ auxiliary space, and $O(n^\epsilon\log^2 n)$ span.

\subsection{Filter, Unstable Partition, and Quicksort}\label{sec:strong-filter}

It is easy to see that we can work on a prefix of the filter problem of size $n^{1-\epsilon}$ using linear work and logarithmic \depth{}, and repeat for $n^\epsilon$ rounds.
The only additional work is to move the unfiltered elements to the beginning of the array, which can be done in linear work and $O(\log n)$ span for each prefix.
This gives a \weak{} algorithm for filter that takes $O(n)$ work,  $O(n^\epsilon\log n)$ span, and $O(n^{1-\epsilon})$ auxiliary space.
We can implement partition similarly, and when moving the unfiltered elements to the beginning, we swap the elements so that at the end of the algorithm, the filtered elements are moved to the end of the array.
This algorithm has the same cost as filter, although the partition result is not stable.
With the \weak{} partition algorithm, we can obtain a \weak{} algorithm for (unstable) quicksort that takes $O(n\log n)$ expected work and $O(n^\epsilon\log^2 n)$ span \whp{}, and $O(n^{1-\epsilon})$ auxiliary space.

\section{Strong PIP Algorithms}
\label{sec:strong}

The \strong{} model is restrictive because of the polylogarithmic
auxiliary space requirement.  To date, only a few non-trivial and
work-efficient \strong{} algorithms have been proposed: reducing and
rotating an array, which are trivial, certain fixed
permutations~\cite{berney2018beyond}, and two-way
partitioning~\cite{kuszmaul2020cache}.  In this section, we
review existing \strong{} algorithms for reduce and rotation, and
present new algorithms for scan (prefix sum), filter, merging, and
sorting.

\subsection{Existing Algorithms}\label{sec:existing}

\myparagraph{Reduce.}
The classic divide-and-conquer algorithm for reduce is already \strong{}.
It is implemented by dividing the input array by two equal sized subarrays, recursively solving the two subproblems in parallel, and finally summing together the partial sums from the two subproblems.
This algorithm requires $O(\log n)$ sequential stack space, $O(n)$ work and $O(\log n)$ span, and so it is an optimal \strong{} algorithm.

\myparagraph{Rotating an array.}
Given an array $[a_1, a_2, ... , a_n]$ and an offset~$o$, the output is a rotated array $[a_{o+1}, \ldots , a_n, a_1, \ldots, a_o]$.
This can be implemented by first reversing $[a_1, \ldots, a_o]$, then reversing $[a_{o+1}, \ldots, a_n]$, and finally reversing the entire array. Reversing can be implemented with a parallel loop, which requires $O(\log n)$ stack space when run serially.
 This algorithm requires $O(n)$ work and $O(\log n)$ span, and is therefore an optimal \strong{} algorithm.

\subsection{Scan}\label{sec:scan}

Scan (prefix sum) is probably the most fundamental algorithmic primitive in parallel algorithm design.
Here we assume $\oplus$ is $+$ (addition) for simplicity, but the results in this section also apply to other associative binary operators.
Non-in-place implementations of scan have been designed since the last century, and the work-efficient version is generally referred to as the Blelloch scan~\cite{blelloch1990pre}.
The Blelloch scan contains two phases.
The first phase is referred to as the ``up-sweep'', which partitions the array into two halves, computes the sum recursively for each half,  then uses the  prefix sums for each half to calculate the prefix sums for the entire sequence, and finally stores this result in auxiliary space.
Then the algorithm applies a ``down-sweep'' phase, which propagates the sums from the first phase down to each element recursively---for a subproblem with a prefix sum of $p$ ($p=0$ for the subproblem corresponding to the whole sequence), we recursively solve the left half with prefix sum $p$, and the right half with prefix sum $p$ plus the sum of the left half, in parallel.
This algorithm takes $O(n)$ work and $O(\log n)$ \depth{}, but unfortunately, it requires linear auxiliary space to store all of the partial sums.

\myparagraph{Making existing approaches in-place.}
We first discuss a solution to make the Blelloch scan in-place.
We partition the array into two equal-sized halves, recursively solve each half, and apply a parallel for-loop to add the sum of the left half to every element in the right half.
Directly applying this algorithm leads to $O(n\log n)$ work, since the recursion tree has $\log_2 n$ levels, and on each level we need to perform $O(n)$ additions, which takes $O(n)$ work and $O(\log n)$ span.
We can reduce the work overhead by stopping the recursion when we reach a subproblem of size no more than $\log_2 n$ (these subproblems constitute the base cases), and apply a sequential in-place scan for these subproblems, and store the partial sums in the last elements of the subproblem arrays.
We then run scan on the sums of the $m=O(n/\log n)$ base cases using the aforementioned algorithm.
This scan takes $O(m\log m)=O(n)$ work and computes the prefix sum before the beginning of each base case.
Lastly, we add this prefix sum to the elements in each base case subproblem to obtain the final result for scan.
This algorithm uses $O(n)$ work, $O(\log^2 n)$ \depth{}, and $O(\log (n/\log n))=O(\log n)$ auxiliary space, which is the recursion depth.

Another approach is to use the Brent-Kung adder~\cite{brent1982regular}, which is a circuit to solve the scan problem with $O(\log n)$ \depth{}, $O(n)$ gates, and $O(n\log n)$ area.
We can change the circuit to an algorithm that contains $O(\log n)$ parallel for-loops and each for-loop simulates the gates at one level.
The work of this algorithm is linear, which is the same as the number of gates, and the \depth{} is $O(\log^2 n)$---$O(\log n)$ parallel for-loops each taking $O(\log n)$ span for forking the tasks.
The output of the original circuit is an inclusive scan (i.e., the output is $[a_1, \ldots , (a_1 \oplus a_2 \oplus \ldots \oplus a_{n})]$). The circuit can be modified to compute the exclusive scan in the same bounds.
In conclusion, we can make the the existing approaches in-place, but their span would not be optimal.

\begin{figure*}[t]
%\vspace{-1em}
\begin{center}
\includegraphics[width=1.75\columnwidth]{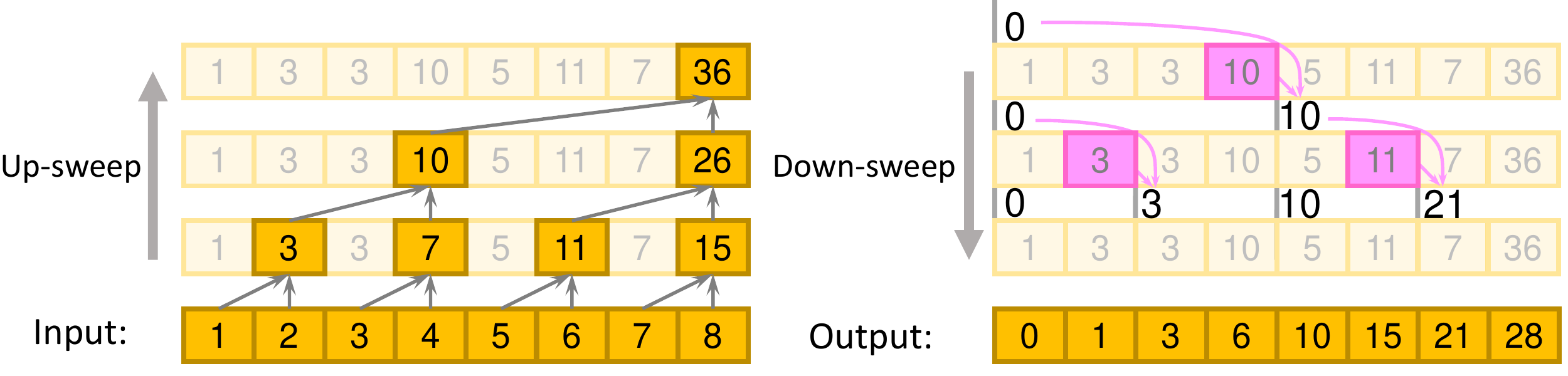}%\vspace{-.5em}
\caption{Our new \strong{} scan algorithm.  It has an up-sweep phase (left) and a down-sweep phase (right).  Each pair of arrows pointing to the same element indicates an addition.}
%\vspace{-1.em}
\label{fig:scan}
\end{center}
\end{figure*}

\begin{algorithm2e}[!t]
\caption{\textsc{In-Place-Scan}}
\label{alg:pip-scan}
\small %\fontsize{9pt}{11pt}\selectfont
%\DontPrintSemicolon
\SetKwProg{inparallel}{In parallel:}{}{}

\KwIn{An array $A_{1\ldots n}$ of size $n$, assuming $A_0=0$.}
\KwOut{The exclusive prefix-sum array of $A$, and sum $\sigma$.}

\smallskip

\textsc{Up-Sweep}$(A, 1, n)$\\
$\sigma\gets A_n$\\
\textsc{Down-Sweep}$(A, 1, n, 0)$\\
\Return {$(A, \sigma)$}

\smallskip

\Fn{\upshape \textsc{Up-Sweep}$(A, s, t)$} {

\lIf{$s=t$} {\Return}

%$k=\lfloor(s+t)/2\rfloor$\\
\inparallel{} {
\textsc{Up-Sweep}$(A, s, \lfloor(s+t)/2\rfloor)$\\
\textsc{Up-Sweep}$(A, \lfloor(s+t)/2\rfloor+1, t)$
}
$A_t\gets A_t+A_{\lfloor(s+t)/2\rfloor}$\\
}
\smallskip

\Fn{\upshape \textsc{Down-Sweep}$(A, s, t, p)$} {

\lIf{$s=t$} {
%\For {$i\gets s$ to $t$} {
%  $x\gets A_i$, $A_i\gets p$, and $p\gets p+x$
%}
  $A_s\gets p$, \Return}
%\lIf{$s\ne 1$} {$A_{\lfloor(s+t)/2\rfloor}\gets A_{s-1}+A_{\lfloor(s+t)/2\rfloor}$\label{line:scan-ass}}
$\mb{LeftSum}=A_{\lfloor(s+t)/2\rfloor}$\\
\inparallel{} {
\textsc{Down-Sweep}$(A, s, \lfloor(s+t)/2\rfloor, p)$\\
\textsc{Down-Sweep}$(A, \lfloor(s+t)/2\rfloor+1, t, p+\mb{LeftSum})$
}
}
\end{algorithm2e}

\myparagraph{A new optimal \strong{} algorithm.}  Our new \strong{}
algorithm is almost as simple as the non-in-place Blelloch scan, and
has the same work and span bounds.  The new algorithm as shown in
Algorithm~\ref{alg:pip-scan}, and illustrated in Figure~\ref{fig:scan}.
In the pseudocode, we assume $A_0=0$ when it is referenced, but the algorithm does not actually need to store this.
The new \strong{} algorithm also contains two phases: the
up-sweep and the down-sweep phases, both of which are recursive.  The
key insight in our new algorithm is to maintain all of the
intermediate results in the input array of $n$ elements, and use stack
space in the down-sweep phase to pass down the partial sums.  For each
recursive subproblem corresponding to a subarray from index $s$ to $t$,
we partition it into two halves, $s$ to $k$ and $k+1$ to $t$, where
$k=\lfloor (s+t)/2\rfloor$.  In the up-sweep phase, we first recursively solve the two subproblems,
and then add the value at index $k$ to the value at index $t$.
These additions are shown
as arrows on the left side of Figure~\ref{fig:scan}.
In the down-sweep phase, we keep the prefix sum $p$ of each subproblem.
Similar to the Blelloch scan, we compute the prefix sum of the right subproblem by adding the sum of the left subproblem to the
current prefix sum.
In our algorithm, the sum of the left subproblem is stored at $A_{\lfloor(s+t)/2\rfloor}$.
Both recursions stop when $i=j$, and at the end of the down-sweep, we obtain the exclusive scan result.
The down-sweep process and its output on an example are shown on the right side of \cref{fig:scan}.

\myparagraph{Correctness and efficiency.}
The correctness and efficiency of this algorithm is based on the following observation.
In the down-sweep phase, the value of $A_t$ in any recursive call is not being used (except for the root where $A_n$ is the total sum).
Hence, in our algorithm, we reuse the space for $A_t$ to store the sum for the next level. The reduction tree (left side of \cref{fig:scan}) has $2n-1$ nodes: $n$ nodes for the input and $n-1$ internal nodes storing the partial sums.
We note that in the down-sweep phase, only the sums of the left subproblems are used, and there are $n-1$ of them.
They are stored in $A_{1,\ldots, n-1}$ by the end of up-sweep, while $A_n$ stores the total sum.
With all of these values, we can run the down-sweep phase in the same way as in the Blelloch scan.
The partial sums stored in $A_{1,\ldots, n-1}$ are passed to the output by the argument  $p$ in the down-sweep function call, which is stored in the stack space.
Hence, the new \strong{} scan algorithm uses $O(n)$ work, $O(\log n)$ \depth{}, and $O(\log n)$ sequential auxiliary stack space, and is therefore an optimal \strong{} algorithm.

\begin{theorem}
The new \strong{} scan algorithm is optimal, using $O(n)$ work, $O(\log n)$ \depth{}, and $O(\log n)$ sequential auxiliary space.
\end{theorem}

\subsection{Other Strong PIP Algorithms}\label{sec:other-strong}

\myparagraph{Filter, unstable partition, and quicksort.}
Consider a $k$-way divide-and-conquer algorithm for filter, where we partition the array into $k$ chunks of equal size, filter each chunk, and pack the unfiltered results together.
For one level of recursion, this takes linear work and $O(k\log n)$ span if chunks are processed one at a time, but within each chunk we move the elements in parallel.
This algorithm only requires a constant amount of extra space to store pointers.
The number of levels of recursion is  $O(\log_k n)$, and so the overall work is $O(n\log n/\log k)$, and the overall span is  $O((k/\log k)\log^2n)$.
Similar to Section~\ref{sec:strong-filter}, we can use this filter algorithm to implement an unstable partition algorithm, with the same cost bounds.
In theory, we can plug in any constant for $k$, which gives a \strong{} algorithm with $O(\log^2 n)$ \depth{} and $O(\log n)$ auxiliary space, although it is not work-efficient.
Alternatively, we can achieve work-efficiency by setting $k=n^\epsilon$. This does not achieve polylogarithmic \depth{}, but has good performance in practice.
We implement this filter algorithm and present experimental results in \cref{sec:exp}.

We can obtain an unstable quicksort algorithm that applies the partition algorithm for $O(\log n)$ levels of recursion \whp{}.
We note that Kuszmaul and Westover~\cite{kuszmaul2020cache} recently developed a work-efficient \strong{} algorithm for partition, which gives a work-efficient and polylogarithmic-span quicksort algorithm.

\myparagraph{Merging and mergesort.}
We again consider merging two sorted arrays of size $n$ and $m$, which are stored consecutively in an array of size $n+m$.
Again, we can use a two-way divide-and-conquer approach, where we use a dual binary search to find the median among all $n+m$ elements, and in parallel swap the out-of-place elements in two arrays.
This swap can be implemented by the \strong{} algorithm for array rotation, discussed in Section~\ref{sec:existing}.
Then, we recursively run merging on the two subproblems, each of size $(m+n)/2$.
The subproblem size shrinks by a factor of 2 on each level of
recursion, and so the recursion depth is bounded by $\log_2 (n+m)$.
The work to swap the elements at each level is $O(n+m)$, and so the
overall work is $O((n+m)\log (n+m))$.  The span and auxiliary space is
$O(\log (n+m))$, which is proportional to the recursion depth.  This gives a \strong{} algorithm for merging. A
\strong{} mergesort algorithm can be obtained by plugging in this merging algorithm, although it is not work-efficient.

\myparagraph{Set operations.}
We now consider computing the union, intersection, and difference of two ordered sets of size $n$ and $m\leq n$.
If the two sets are given in a binary tree format, then existing algorithms for these operations~\cite{BlellochFS16,pam} are already \strong{}, work-optimal ($O(m\log (n/m+1))$ work), and have $O(\log^2 n)$ span. 
We now describe how to implement these operations if the sets are given in arrays stored contiguously in memory.
For union, we can first use the merging algorithm described above, and then the filter algorithm described above to remove duplicates. Therefore, computing the union on arrays is \strong{}.
For intersection and difference, we can
run binary searches to find each element in the smaller set inside the larger set, and then apply the filter algorithm described above to obtain the output, which takes $O(n\log n)$ work.
The resulting algorithms are not work-efficient, since our  \strong{} merging and filter are not work-efficient.

\section{Relaxed PIP Graph Algorithms}
\label{sec:relax}

In this section, we introduce new \weak{} algorithms for graph
connectivity, biconnectivity, and minimum spanning forest.

\myparagraph{Connectivity and Biconnectivity.}
The standard output size for graph connectivity and biconnectivity is $O(n)$ and $O(m)$, respectively.
Recent work by Ben-David et al.~\cite{bendavid2017implicit} introduces a compressed scheme for storing graph connectivity information.
For any $1\le k\le n$, it requires $O(k\log n+m/k)$ output size with an $O(k)$ expected query work for connectivity and $O(k^2)$ expected query work for biconnectivity.
Constructing such a compressed (bi)connectivity oracle takes $O(km)$ expected work and $O(k^{3/2}\log^3n)$ span \whp{}.
By setting $k=m^\epsilon$, we have have the following theorem.
\begin{theorem}\label{thm:biconnectivity}
  A (bi)connectivity oracle can be constructed using $O(m^{1-\epsilon})$ auxiliary space, $O(m^{1+\epsilon})$ expected work, and $O(m^{3\epsilon/2}\log^3n)$ span \whp{} for $1/2<\epsilon<1$.
  A connectivity query can be answered in $O(m^\epsilon)$ expected work, and a biconnectivity query can be answered in $O(m^{2\epsilon})$ expected work.
\end{theorem}

The high-level idea in the algorithms is to select a subset of the vertices as the ``centers'' and only keep information for these center vertices. Each vertex has a   $1/k$ probability of being selected as a center.
This is referred to as the \defn{implicit decomposition} of the graph.
For a query to a non-center vertex $v$, we apply a breadth-first search from $v$ to the first center $c$, which takes $O(k)$ expected work~\cite{bendavid2017implicit}.
For connectivity, $v$'s label is the same as $c$'s label.
It is also possible that a search does not reach any center, but Ben-David et al.~\cite{bendavid2017implicit} show that the expected size of a connected component without a center vertex is small ($O(k)$ in expectation), and so the cost to traverse all vertices in such a component is also $O(k)$ in expectation.
For biconnectivity, an additional step of local analysis is required to obtain the output for $v$ from $c$, which requires $O(k^2)$ expected work.

Theorem~\ref{thm:biconnectivity} gives algorithms that are almost \weak{}, other than having an extra factor of $O(m^{\epsilon/2})$ in the product of the space and span bounds. 
Alternatively, we can obtain new \weak{} connectivity and biconnectivity algorithms by using the minimum spanning forest algorithm that will be discussed next, at a cost of additional work.

\myparagraph{Minimum Spanning Forest.}  The idea of implicit
decomposition can be extended to the minimum spanning forest (MSF)
problem.  For simplicity, we assume that the graph is connected, but
disconnected graphs can also be handled using an approach described by
Ben-David et al.~\cite{bendavid2017implicit}.

We note that the MSF is unique
for a graph (assuming that ties are broken consistently).  Therefore,
for a query to vertex $v$, instead of using a breadth-first search
on \emph{all} edges to find the center in connectivity, we need to search out to a
center using \emph{only} the MSF edges.  This can be achieved by
using a Prim-like search algorithm from $v$.
This increases the work by a factor of $O(\log k)$ to compute the
implicit decomposition of the graph and for the query cost (the queue will contains $O(k)$
vertices on average for each search).

We can generate an implicit decomposition of the graph using a similar approach as for connectivity and biconnectivity.
We then compute the MSF
across the  $m/k$ centers of the decomposition.
The output size of this spanning forest is $O(m/k)$.  To
compute the MSF in parallel, we can use
\Boruvka{} algorithm.
We start with every cluster being in its own component, and enumerate all edges for $O(\log n)$ rounds until
the entire graph is connected.  On each round, we run \Boruvka{}
algorithm to find the minimum outgoing edges from each component.
This takes $O(mk\log k)$ work---we check all $m$ edges in a \Boruvka{} round and each edge takes $O(k\log k)$ work to find the clusters of both of its endpoints.
Similar to connectivity, each vertex only uses MST edges to reach the centers, and so the algorithm based on implicit decomposition is correct.
By setting $k=m^\epsilon$, the cost for each round is $O(m^{1+\epsilon}\log n)$ work and $O(k\log k+\log n)=O(\log n+m^\epsilon\log m)$ span.
Since there are $O(\log n)$ rounds, we obtain the following theorem.

\begin{theorem}\label{thm:msf}
  Given a graph with $n$ vertices and $m$ edges, a data structure for minimum spanning forest can be computed in $O(m^{1+\epsilon}\log^2 n)$ expected work, $O(m^\epsilon\polylog(n))$ span \whp{}, and $O(m^{1-\epsilon})$ auxiliary space.
  Querying if an edge is in the MSF takes $O(m^{1-\epsilon}\log n)$ work.
\end{theorem}

This MSF algorithm is \weak{}.  
Once we have the implicit spanning forest (the minimum-weight one in this case), we can use the same approaches mentioned above to get \weak{} algorithms
for connectivity and biconnectivity.
Compared to Theorem~\ref{thm:biconnectivity}, the MSF-based algorithms require a factor of $O(\log^2 n)$ more work, but have lower span.

\myparagraph{Related Work.}
Many researchers have studied the
time-space tradeoff for the $s$-$t$ connectivity problem, and the
results lead to in-place algorithms for the
problem~\cite{broder1994trading,feige1997spectrum,beame1998time,kosowski2013faster,edmonds1998time}.
However, the algorithms are based on random walks and are inherently
sequential.  Recent work by Chakraborty et
al.~\cite{Chakraborty2018,Chakraborty2019,Chakraborty2020} has studied
in-place algorithms for other graph problems, including graph search
and connectivity, and it would be interesting to parallelize these
algorithms in the future.

\section{Implementations and Experiments}\label{sec:exp}

In the previous sections, we have designed parallel in-place algorithms
with strong theoretical guarantees.  Many of these algorithms are
relatively simple, and
in this section we describe how to implement these algorithms
efficiently so that they can outperform or at least be competitive
with their non-in-place counterparts, while using less space.  We present implementations for
five algorithms: scan, filter, random permutation, list contraction,
and tree contraction.  The implementations for the first two are
fairly simple, and the last three are based on the deterministic
reservations framework of Blelloch et
al.~\cite{BlellochFinemanGibbonsEtAl2012}.

\subsection{Experimental Setup}

We run all of our experiments on a 72-core Dell PowerEdge R930 (with
two-way hyper-threading) with 4$\times$2.4GHz Intel 18-core E7-8867 v4
Xeon processors (with a 4800MHz bus and 45MB L3 cache) and 1TB of main
memory.  We compile the code using the \texttt{g++} compiler (version
5.4.1) with the \texttt{-O3} flag, and use Cilk Plus for parallelism.

We compare our PIP algorithms to the non-in-place versions in
the Problem Based Benchmark Suite (PBBS)~\cite{shun2012brief}, which is a collection of highly-optimized
parallel algorithms and implementations and widely used in
benchmarking.  The implementations of random permutation, list
contraction, and tree contraction in PBBS are
from~\cite{shun2015sequential}.

\begin{figure*}[h]
%\vspace{-1em}
%\vspace{10em}
\begin{minipage}[t]{0.24\textwidth}
\begin{center}
  Scan\\
  \includegraphics[width=\columnwidth]{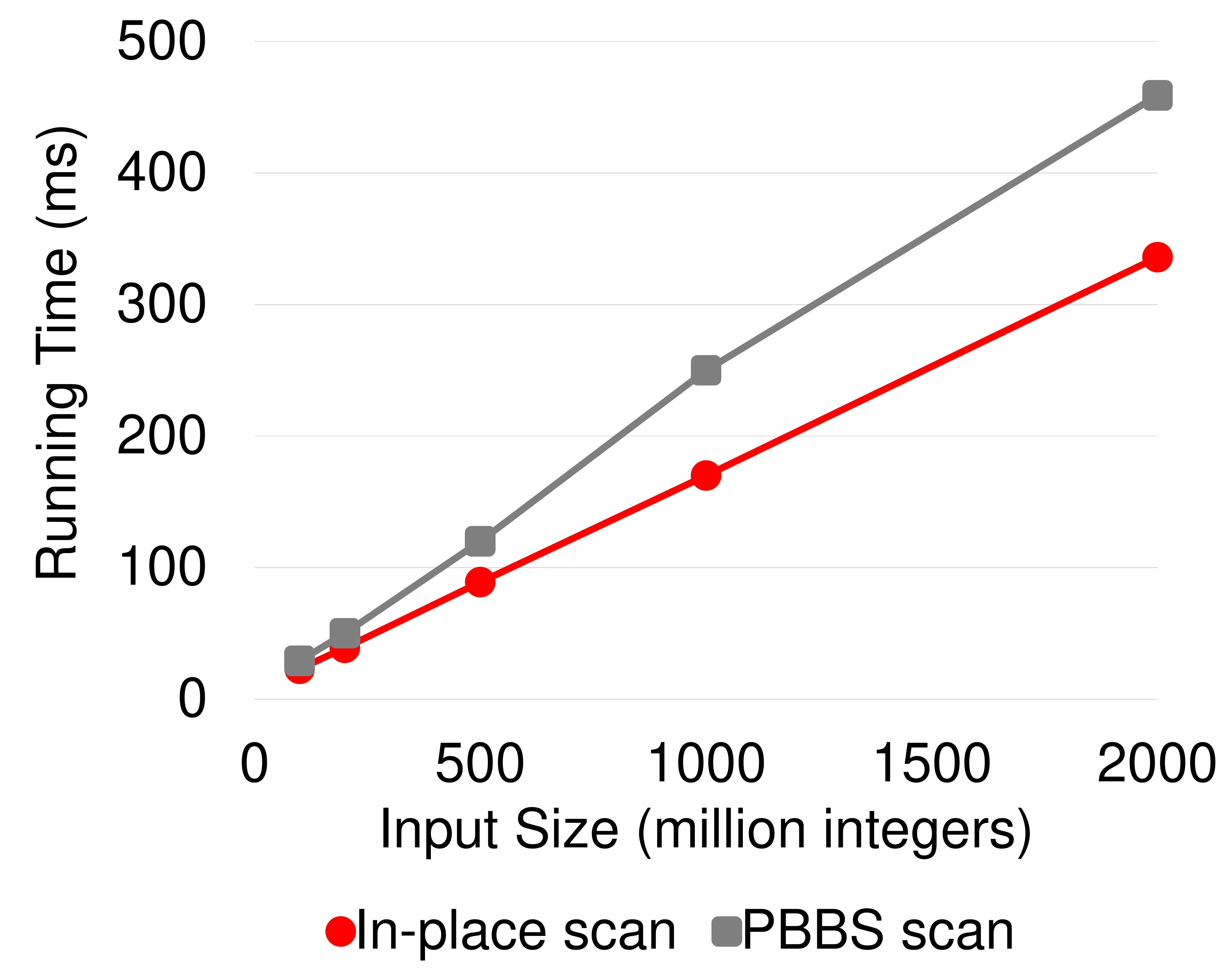}
  \includegraphics[width=\columnwidth]{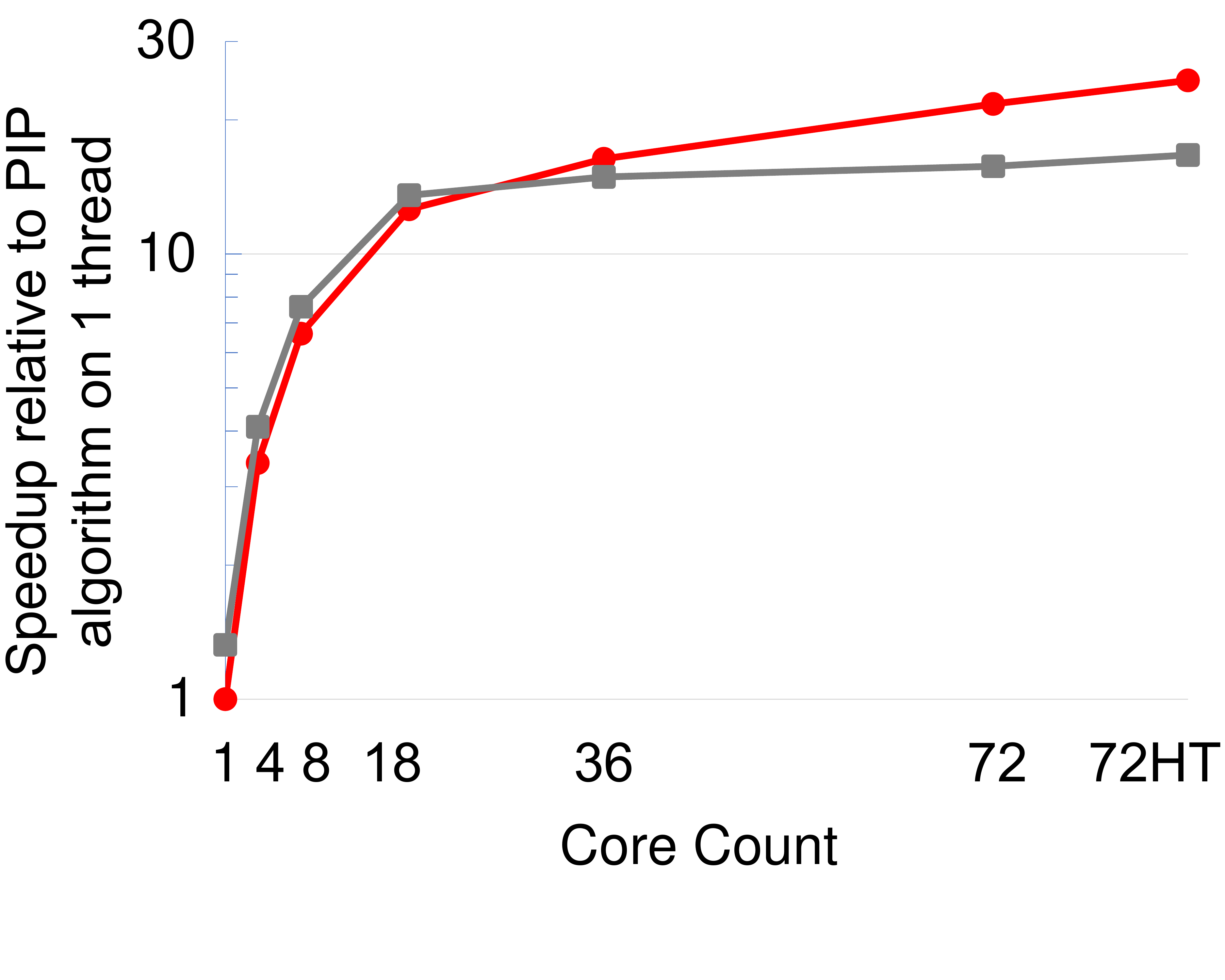}
\end{center}%\vspace{-1.2em}
\end{minipage}
\begin{minipage}[t]{0.24\textwidth}
\begin{center}
  Filter\\
  \includegraphics[width=\columnwidth]{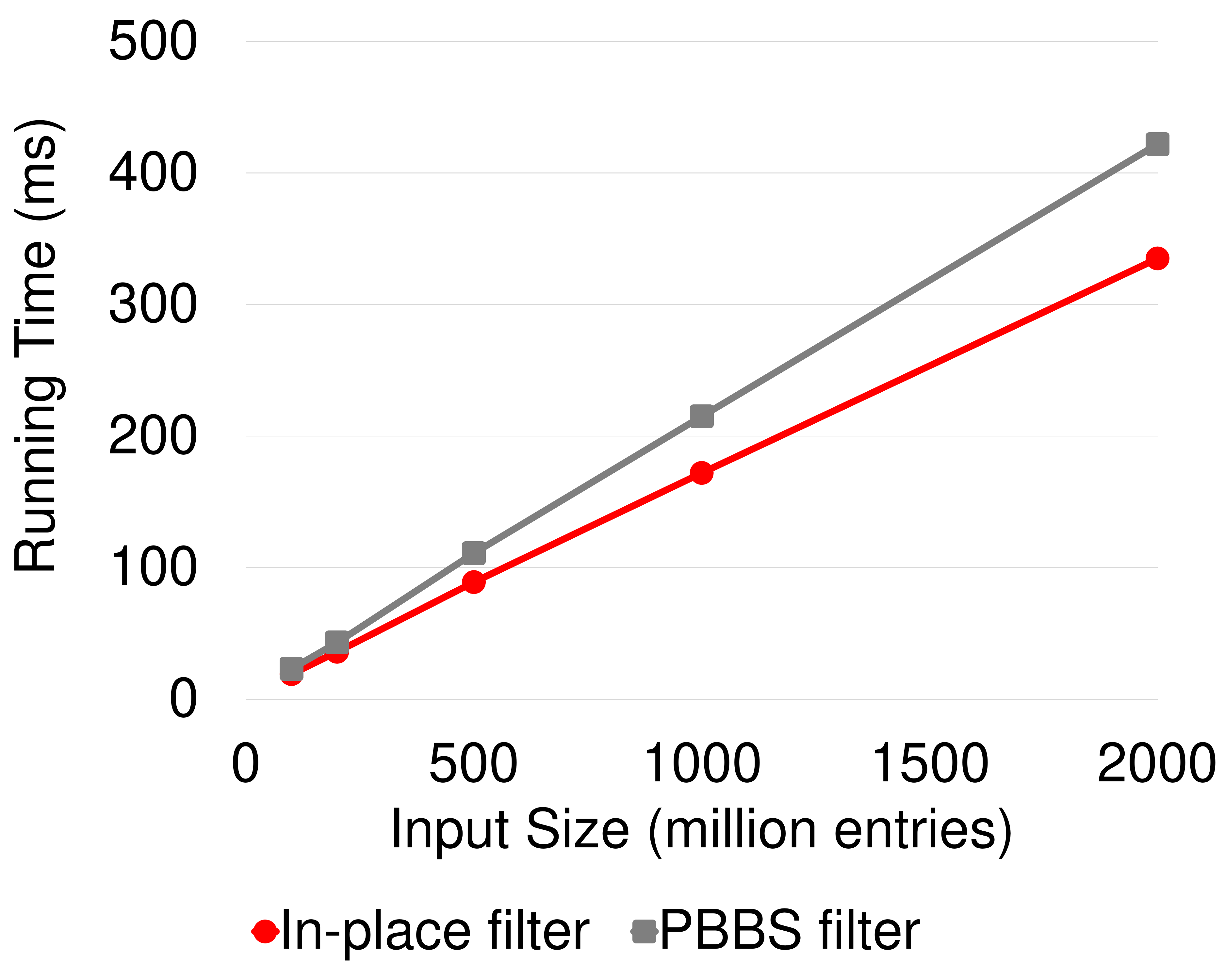}
  \includegraphics[width=\columnwidth]{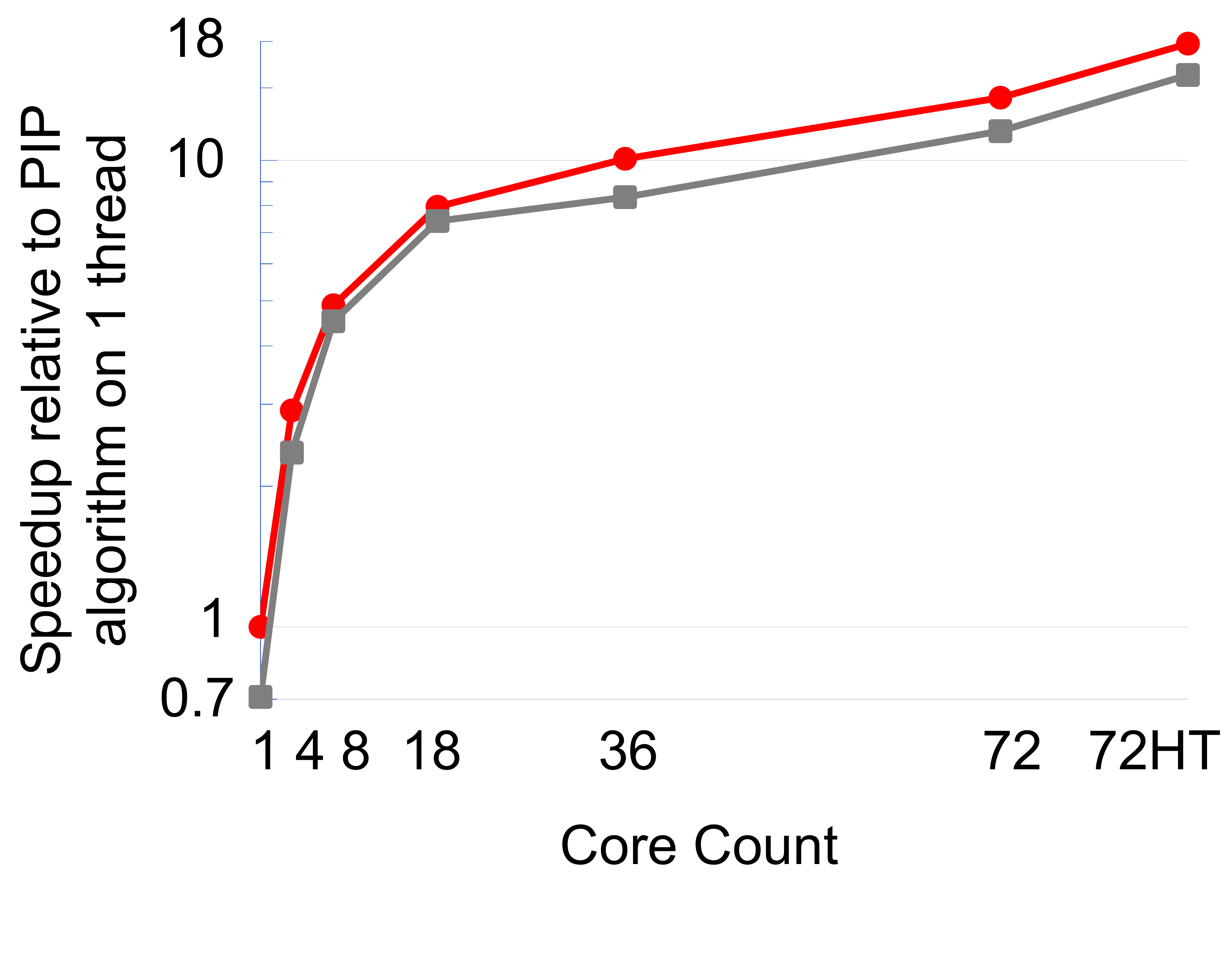}
\end{center}%\vspace{-1.2em}
\end{minipage}
\begin{minipage}[t]{0.24\textwidth}
\begin{center}
  List Contraction\\
  \includegraphics[width=\columnwidth]{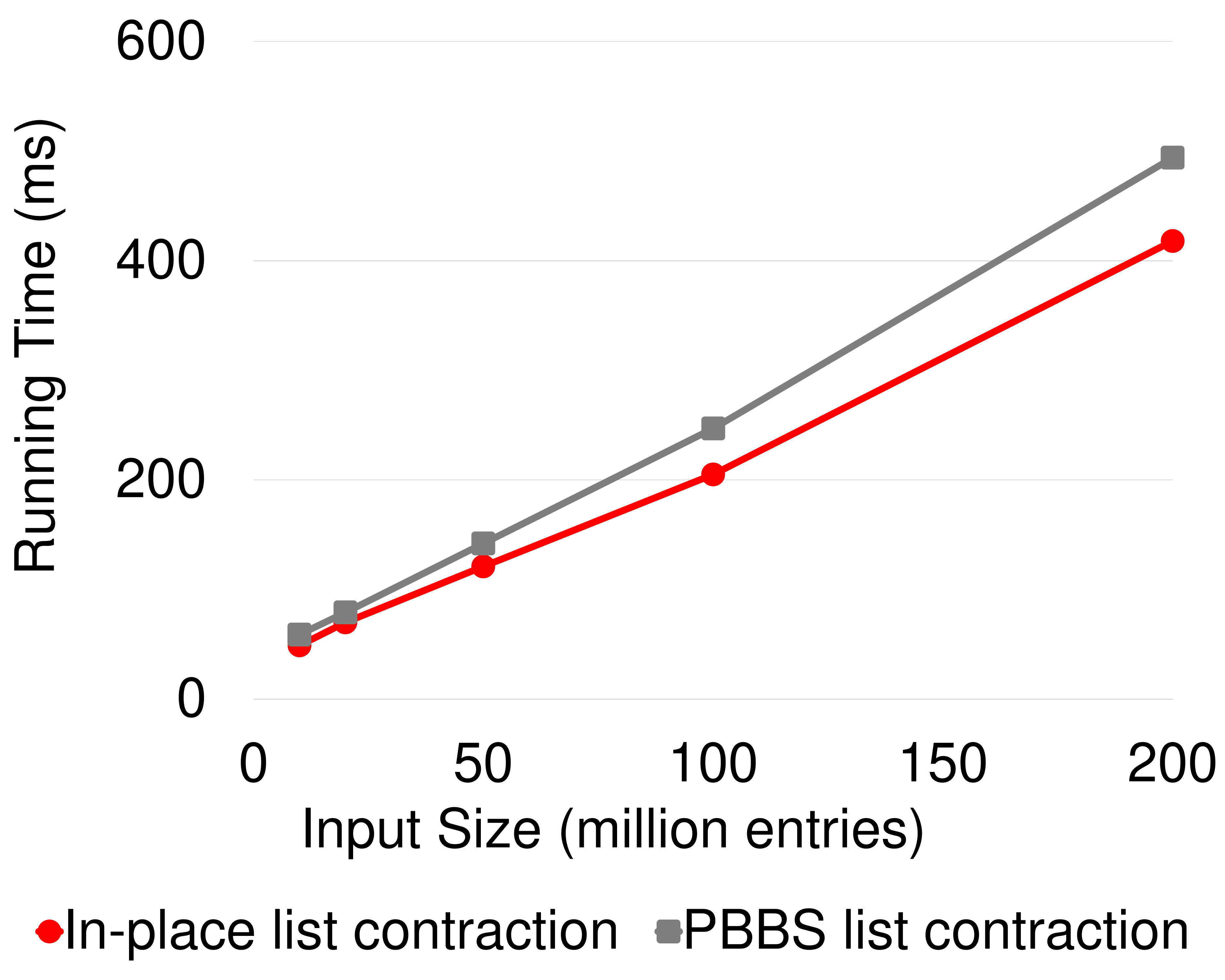}
  \includegraphics[width=\columnwidth]{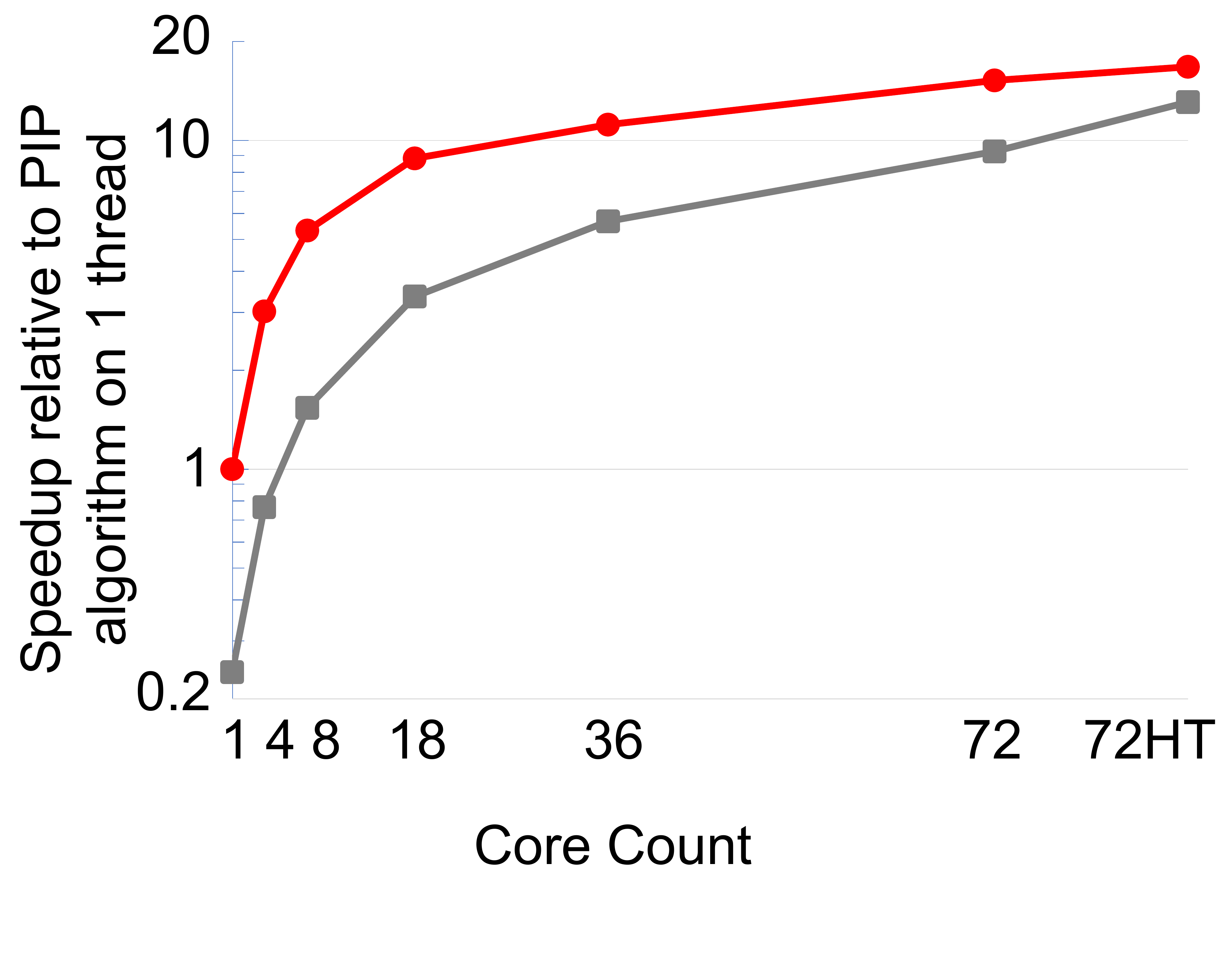}
\end{center}%\vspace{-1.2em}
\end{minipage}
\begin{minipage}[t]{0.24\textwidth}
\begin{center}
  Tree Contraction\\
  \includegraphics[width=\columnwidth]{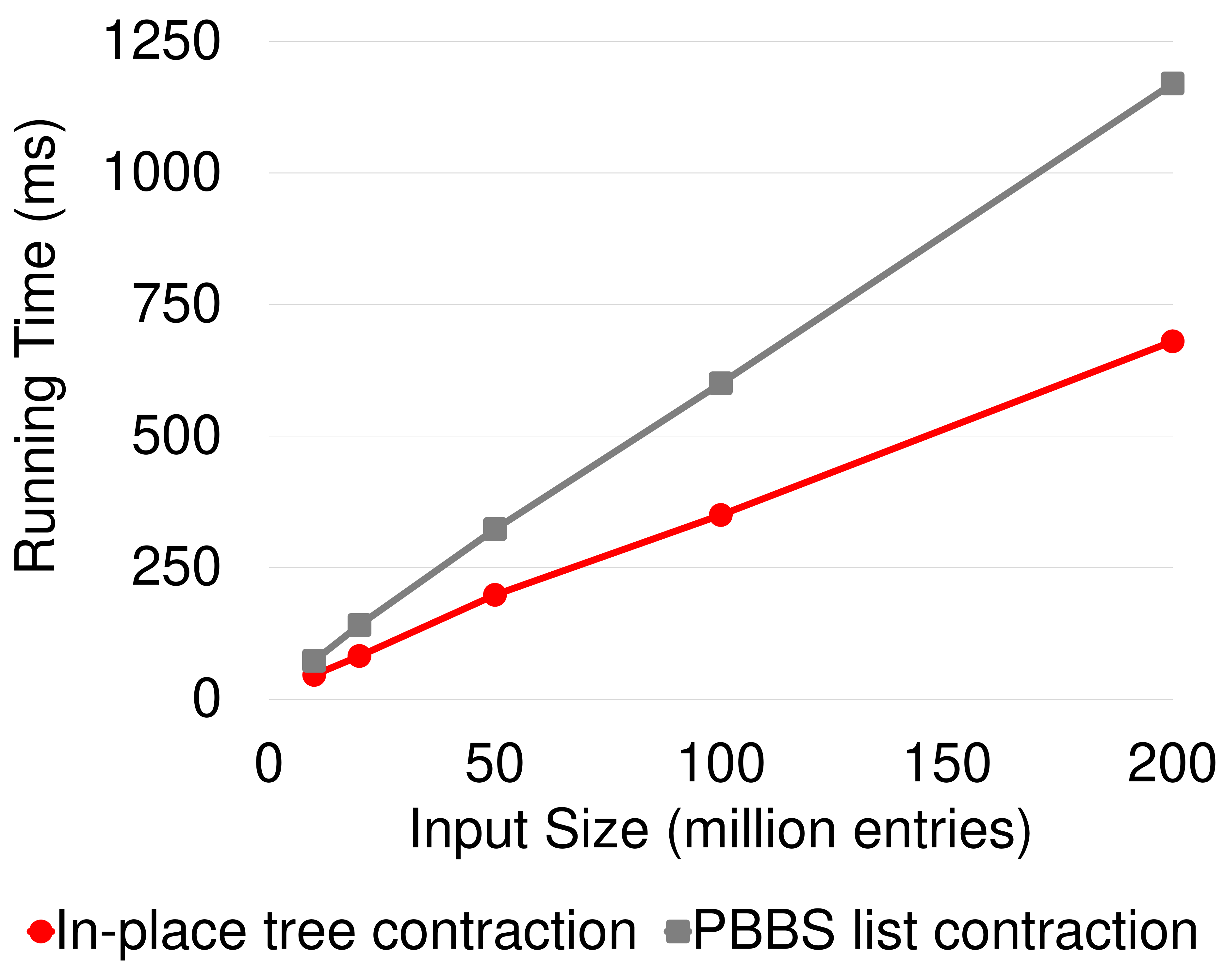}
  \includegraphics[width=\columnwidth]{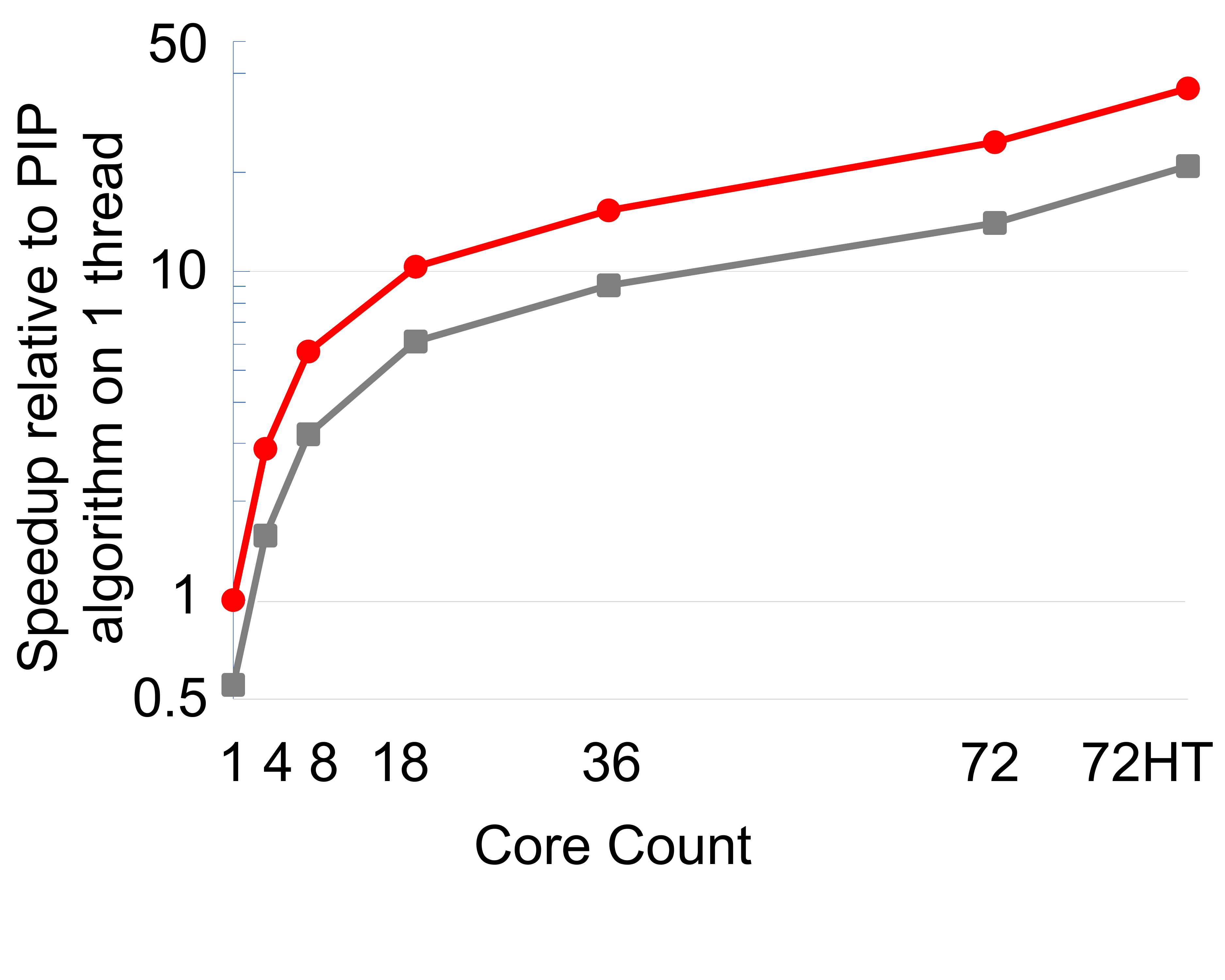}
\end{center}%\vspace{-1.2em}
\end{minipage}
\vspace{-1em}
\caption{ Running times and speedups for scan, filter, list contraction, and tree contraction.
The top figures show the running times of our parallel in-place algorithms and their non-in-place counterparts from PBBS for varying input sizes.
The bottom figures show the parallel speedups compared to the parallel PIP implementations on 1-thread, and the core count varies from 1 core to all 72 cores with two-way hyper-threading (72HT).}
%In both cases the relative performance between in-place scan and PBBS scan is consistent.}
%\vspace{-2em}
\label{fig:exp}
\end{figure*}

\hide{\vspace{-1em}
\begin{center}
  Scan (scalability for 1000 million integers)\\
  \includegraphics[width=.45\columnwidth]{data/scan-time}
  \includegraphics[width=.7\columnwidth]{data/scan-scale}\\
  Filter (scalability for 1000 million entries)\\
  \includegraphics[width=.7\columnwidth]{data/filter-time}
  \includegraphics[width=.7\columnwidth]{data/filter-scale}\\
  List Contraction (scalability for 100 million integers)\\
  \includegraphics[width=.7\columnwidth]{data/list-time}
  \includegraphics[width=.7\columnwidth]{data/list-scale}\\
  Tree Contraction (scalability for 100 million integers)\\
  \includegraphics[width=.7\columnwidth]{data/tree-time}
  \includegraphics[width=.7\columnwidth]{data/tree-scale}\\}

\subsection{Scan and Filter}

For scan, we implement Algorithm~\ref{alg:pip-scan} and switch to a
sequential in-place scan when the subproblem size is less than 256.
For filter, we implement the PIP algorithm from
Section~\ref{sec:other-strong}, but we keep the implementation
work-efficient by setting the branching factor $k=\sqrt{n}$, and only apply
one level of recursion.  However, this increases the span to
$O(\sqrt{n}\log n)$ and has $O(\sqrt{n})$ rounds of global
synchronization (the $k$ chunks are processed one after another), which is a significant overhead.  We use the
following optimization to significantly reduce this overhead in practice.
We move the elements from multiple consecutive chunks in parallel as long as the destination of the
last chunk is before the original location of the first chunk.  We
apply a binary search in each round to find the maximum number of
chunks that can be moved in parallel.  If the unfiltered elements
  are distributed relatively evenly in the input and the output size is a constant fraction of the input, then the
algorithm requires logarithmic rounds to finish.

We compare our PIP algorithms to the non-in-place versions in
PBBS. The PBBS scan is the classic Blelloch scan
implementation~\cite{Blelloch89} and the filter is similar to our
implementation, but the output is stored in a separate
array.
In the PBBS filter, it first filters
each $\sqrt{n}$-sized chunk in parallel while each chunk is
  processed sequentially, and then moves the remaining elements to a
separate output array in parallel.

The running times and scalability (parallel speedup relative to the best algorithm on 1 thread, which was our PIP algorithm in all cases) for
scan and filter on different input sizes are shown in
Figure~\ref{fig:exp} and Tables~\ref{tab:scan-filter}--\ref{tab:scan-filter2}.  For filter, 50\% of the input entries are kept
in the output.  Our in-place scan is 30--45\% faster and our in-place
filter is about 25--30\% faster than their non-in-place counterparts
due to having a smaller memory footprint.  The speedups
are also competitive or better than the non-in-place versions.  For
filter, the fraction of elements in the output affects the
performance of both our algorithm and the PBBS algorithm. A larger output fraction increases
the number of rounds for movement and global synchronization in the PIP filter
algorithm.  In \cref{tab:filter-fraction}, we vary the output fraction
and show that our new algorithms range from about 2x faster (12.5\%
output) to having about the same performance (87.5\% output).

\begin{table}[t]
  %\small
  \centering
    \begin{tabular}{crrrr}
    \toprule
    \multicolumn{1}{c}{Input Size} & \multicolumn{2}{c}{Scan} & \multicolumn{2}{c}{Filter} \\
    \multicolumn{1}{c}{(million)} & \multicolumn{1}{c}{PBBS} & \multicolumn{1}{c}{PIP} & \multicolumn{1}{c}{PBBS} & \multicolumn{1}{c}{PIP} \\
    \midrule
    100   & 29    & 23    & 23    & 19 \\
    200   & 50    & 39    & 43    & 36 \\
    500   & 120   & 89    & 111   & 89 \\
    1000  & 250   & 170   & 215   & 172 \\
    2000  & 459   & 336   & 422   & 335 \\
    \bottomrule
    \end{tabular}%
  \caption{Running times (in milliseconds) of the PBBS algorithms and our new PIP algorithms for scan and filter on 72 cores with hyper-threading. }
  \label{tab:scan-filter}%
    \begin{tabular}{crrrr}
    \toprule
    \multicolumn{1}{c}{Core} & \multicolumn{2}{c}{Scan} & \multicolumn{2}{c}{Filter} \\
    \multicolumn{1}{c}{count} & \multicolumn{1}{c}{PBBS} & \multicolumn{1}{c}{PIP} & \multicolumn{1}{c}{PBBS} & \multicolumn{1}{c}{PIP} \\
    \midrule
    1     & 3150  & 4170  & 2770  & 2150 \\
    4     & 1020  & 1230  & 861   & 701 \\
    8     & 548   & 630   & 524   & 393 \\
    18    & 308   & 331   & 301   & 244 \\
    36    & 280   & 255   & 254   & 191 \\
    72    & 265   & 192   & 222   & 179 \\
    72HT  & 250   & 170   & 215   & 172 \\
    \bottomrule
    \end{tabular}%
  \caption{Running times (in milliseconds) of the PBBS algorithms and our new PIP algorithms for scan and filter on varying core counts.}
  \label{tab:scan-filter2}%
\end{table}%

\begin{table}[t]%
  %\small
  \centering
   \begin{tabular}{cccccc}
    \toprule
    Output fraction &  12.5\% & 25\%   & 50\%   & 75\% & 87.5\% \\
    PBBS filter & 189   & 202   & 215   & 234 & 252 \\
    PIP filter  & 94    & 118   & 172   & 212 & 254 \\
    \bottomrule
    \end{tabular}%
  \caption{Running times (in milliseconds) of the PBBS filter algorithm and our new PIP filter algorithm, with varying output fraction on 72 cores with hyper-threading.  The input is 1 billion integers.\label{tab:filter-fraction}}
\end{table}%

This experiment indicates that using the PIP scan and filter algorithms can improve
both the running time and memory usage over the non-in-place filter
algorithm, and is preferable when the input can be overwritten.
We note that ParlayLib~\cite{blelloch2020parlaylib}, the latest version of PBBS, also includes the in-place versions of scan and filter, and we plan to compare with these in the future.

\subsection{Deterministic Reservations}

Implementing the PIP algorithms for random permutation, list
contraction, and tree contraction is more challenging since they are more complicated than scan and filter.
However, the \dproperty{} can greatly simplify the implementation of these
algorithms, and we only need to design an efficient implementation
working on a prefix of the problem and run it iteratively.
Interestingly, the original implementations of these algorithms
in~\cite{shun2015sequential} are based on a framework named as the
\emph{deterministic reservations}~\cite{BlellochFinemanGibbonsEtAl2012} that
runs similarly in rounds, where each round processes a prefix of the
remaining elements.
We  first briefly overview the framework of deterministic reservations, and then discuss how we can modify the original
implementations to obtain new PIP algorithms for random permutation,  list contraction, and tree contraction.

Deterministic reservations
is a framework for iterates in a parallel algorithm to check if all of
their dependencies have been satisfied through the use of shared data structures,
and executing the ones that have been
satisfied~\cite{BlellochFinemanGibbonsEtAl2012}.  Deterministic reservations
proceeds in rounds, where on each round, each remaining iterate tries
to execute.  Iterates that fail to execute will be packed and
processed again in the next round.  To achieve good performance in
practice, instead of processing all iterates on every round, the
framework only works on a prefix of the remaining iterates (usually
around 1--2\% of the input iterates).
This naturally meets our requirement for controlling the number of elements to process in the \weak{} algorithms. We provide more details on the framework in the appendix.

\subsection{Random Permutation}

We implement the PIP random
permutation algorithm (Algorithm~\ref{alg:parallel-rp}) based on
deterministic reservations.  We have four implementations (\naive{},
\rpflat{}, \oneres{}, and \final{}), and each one improves upon the
previous one.
We compare our implementation to the code in PBBS library, and here we refer to it as \algname{RP-PBBS}. More details on our implementations are provided in the appendix.

\algname{RP-PBBS} runs in rounds, where each round processes a prefix of 2\% of the total number of  swaps.
This approach naturally fits with our PIP random permutation algorithm.
Our PIP algorithms also process 2\% of the total number of swaps, which empirically gave the best performance.

  \begin{table}[t]
    %\small
  \centering
    \begin{tabular}{lccc}
    \toprule
    Phase: & Reserve  & Commit  & Cleaning  \\
    \midrule
    \algname{RP-PBBS}  & 2/1  & 2/1  & 0/0 \\
    \naive{} & 0/3  & 0/3  & 3/0 \\
    \rpflat{} & 1/2  & 1/2  & 2/0 \\
    \oneres{} & 1/1& 1/2  & 1/0 \\
    \final{} & 1/1  & 2/1  & 1/0  \\
    \bottomrule
    \end{tabular}%
  \caption{ Approximate number of sequential/random access per swap on each round for the five implementations of random permutation.}
  \label{tab:rp-count}%
  \end{table}

The overall goal in our implementations is to reduce the number of
memory accesses in the algorithm.  The original \algname{RP-PBBS}
implementation from PBBS needs roughly 4 sequential accesses and 2
random accesses per swap on each round.  In our new PIP algorithm, we
need a data structure to hold all associated memory accesses in the
prefix for auxiliary arrays $R$ and $H$.  In \naive{}, we simply use
concurrent hash tables~\cite{shun2014phase}, and this implementation incurs roughly 3
sequential accesses and 6 random accesses per swap on each round.  As
an improvement, \rpflat{} uses an array to replace the hash table for the $H$ array, which changes the number of sequential and random
accesses to 4 and 4, respectively.  \oneres{} removes one of  the reservations
(the one on Line~\ref{line:resend} of Algorithm~\ref{alg:parallel-rp}) and reduces
the number of random accesses to 3. 
Our final version, \final{}, uses an array
instead of a hash table for the part of $R$ that is accessed
contiguously by the iterates.  \final{} incurs 4 sequential accesses
and 2 random accesses per swap on each round, which is the same as in
the non-in-place version.  The approximate numbers of sequential and random
memory accesses for each implementation are given in \cref{tab:rp-count}.

\begin{figure*}
%\vspace{1em}
\begin{center}
  \includegraphics[width=.9\columnwidth]{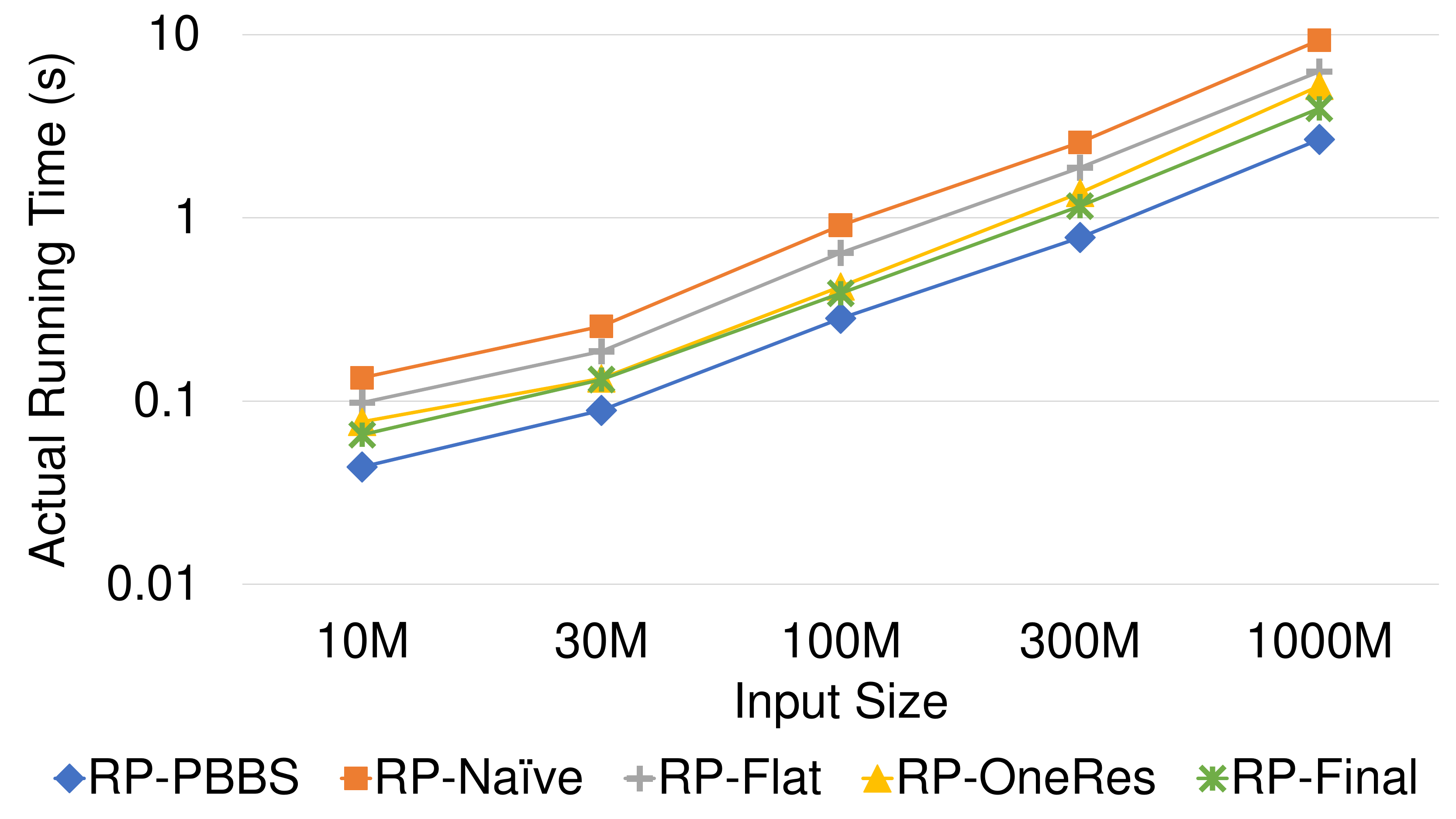}
  \includegraphics[width=.9\columnwidth]{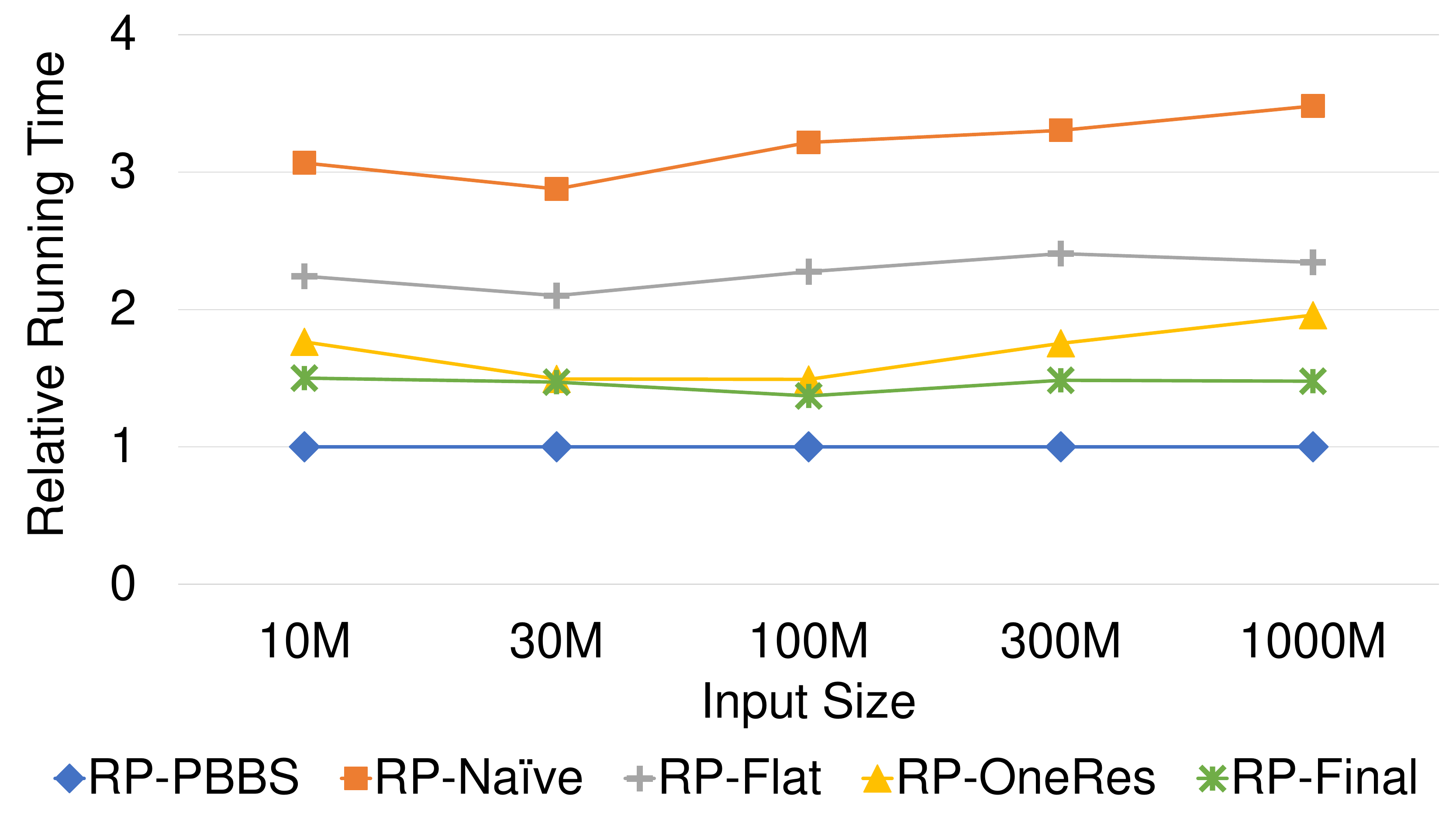}
\end{center}%\vspace{-1em}
\caption{ The actual running times (left) and running times relative to \algname{RP-PBBS} (right) of different implementations for random permutation on  72 cores with hyper-threading.  The input size varies from 10 million to 1 billion integers.
}
%\vspace{-2em}
\label{fig:rp-exp}

%\vspace{11em}
\end{figure*}

We test the performance of our implementations on inputs of size 10 million
to 1 billion 64-bit integers, and compare them with the best non-in-place
counterpart, which is from PBBS (\algname{RP-PBBS}).  The actual running times are shown in Figure~\ref{fig:rp-exp} and
\cref{tab:rp-runtime}.
In Figure~\ref{fig:rp-exp} (left), we see that all of the
implementations have similar and consistent scalability with respect
to input size. In Figure~\ref{fig:rp-exp} (right), we show the
running times relative to \algname{RP-PBBS}.
\final{} only has a modest overhead of 30--40\%
over \algname{RP-PBBS}, while only using 4\% of the auxiliary space
required by \algname{RP-PBBS}.
We can further reduce
the additional space by shrinking the prefix size, at the expense of
having more rounds, and hence more global synchronization.  In
Table~\ref{tab:rp-additional-space}, we present the running times under different amounts of additional space for \algname{RP-Final}.

\begin{table}[t]
  %\small
    \centering
    \begin{tabular}{lrrrrr}
    \toprule
       Input size: & \multicolumn{1}{l}{10M} & \multicolumn{1}{l}{30M} & \multicolumn{1}{l}{100M} & \multicolumn{1}{l}{300M} & \multicolumn{1}{l}{1000M} \\
    \midrule
    \algname{RP-PBBS}  & 43.7  & 89.2  & 283   & 781   & 2680 \\
    \naive{} & 134   & 256   & 910   & 2580  & 9330 \\
    \rpflat{} & 98.1    & 187   & 644   & 1880  & 6280 \\
    \oneres{} & 77.1  & 133   & 422   & 1370  & 5250 \\
    \final{} & 65.6  & 131   & 388   & 1160  & 3960 \\
    \bottomrule
    \end{tabular}%
  \caption{Running time (in milliseconds) of the five implementations of random permutation on 72 cores with hyper-threading.}
  \label{tab:rp-runtime}%
  \end{table}

%\vspace{-1em}
  \begin{table}[t]
  \centering    %\small
    \begin{tabular}{lcccc}
    \toprule
    Additional space &  0.4\% & 1\%   & 2\%   & 4\% \\
    Running time (ms) & 537   & 425   & 411   & 388 \\
    \bottomrule
    \end{tabular}%
%  \vspace{1em}
  \caption{Running time with different restrictions on additional space for \final{} on 72 cores with hyper-threading.  The input is 100 million 64-bit integers.\label{tab:rp-additional-space}}
  \end{table}

\subsection{List Contraction and Tree Contraction}

Similar to random
permutation, for list contraction and tree contraction, we design our PIP
implementations based on the non-in-place implementations from
PBBS~\cite{shun2012brief,shun2015sequential}. In this case, the auxiliary space
in list contraction and tree contraction is the $R$ array (shown in Algorithm~\ref{algo:lc}), which has linear size.
In list contraction (\cref{algo:lc}) and tree contraction, $R[i]$ represents whether node $i$ can be
contracted in the current round.  The deterministic reservations framework also
needs to keep this information stored in another form, since it needs to pack the remaining
iterates for the next round.
We optimized the implementations to remove the $R$ array and store the information directly in the  deterministic reservations framework---instead of indicating if the $i$'th iteration can be contracted (using the $R$ array), we use an array in the framework that stores this information only for
the iterations in the current prefix.
Hence, our implementation does not need to explicitly store the $R$ array, and thus only needs storage proportional to the prefix size.

\begin{table}[t]
  %\small
  \centering
  \centering
    \begin{tabular}{crrrrrr}
    \toprule
    Input Size & \multicolumn{3}{c}{List contraction} & \multicolumn{3}{c}{Tree contraction} \\
    (million) & &\multicolumn{1}{c}{PBBS} & \multicolumn{1}{c@{}}{PIP} & & \multicolumn{1}{c}{PBBS} & \multicolumn{1}{c@{}}{PIP} \\
    \midrule
    10   & & 59    & 49   & & 73    & 46 \\
    20   & & 79    & 70   & & 141   & 82 \\
    50   & & 142   & 131  & & 323   & 198 \\
    100  & & 247   & 205  & & 600   & 350 \\
    200  & & 494   & 418  & & 1170  & 680 \\
    \bottomrule
    \end{tabular}%
  \caption{Running times (in milliseconds) of the PBBS algorithms and our new PIP algorithms for list and tree contraction on 72 cores with hyper-threading.}
  \label{tab:lt-runtime}%
  %\small
    \begin{tabular}{crrrrrr}
    \toprule
    Core & \multicolumn{3}{c}{List contraction} & \multicolumn{3}{c}{Tree contraction} \\
    count & &\multicolumn{1}{c}{PBBS} & \multicolumn{1}{c@{}}{PIP} & & \multicolumn{1}{c}{PBBS} & \multicolumn{1}{c@{}}{PIP} \\
    \midrule
    1    & & 13900 & 3350  & & 22800 & 12600 \\
    4    & & 4370  & 1110  & & 8020  & 4370 \\
    8    & & 2180  & 636   & & 3950  & 2210 \\
    18   & & 1010  & 384   & & 2060  & 1220 \\
    36   & & 592   & 302   & & 1390  & 1050 \\
    72   & & 362   & 220   & & 899   & 512 \\
    72HT & & 247   & 205   & & 603   & 350 \\
    \bottomrule
    \end{tabular}%
  \caption{Running times (in milliseconds) of the PBBS algorithms and our new PIP algorithms for list and tree contraction on varying core counts.}
  \label{tab:lt-scale}%
  \end{table}

We test the performance of our implementations on inputs of between 10 million to 200 million entries, each of which contain two 64-bit pointers.
The running times and speedups over the the PIP algorithm on 1 thread as a function of core counts are shown in \cref{fig:exp} and Tables~\ref{tab:lt-runtime}--\ref{tab:lt-scale}.
Since we eliminated the use of the $R$ array in our PIP algorithms and simplified the logic in the implementation (which required changes to the deterministic reservations framework), the parallel execution time improved by 15--20\% for list contraction and 60--70\% for tree contraction  compared to non-in-place version in PBBS by Shun et al.~\cite{shun2015sequential}.
It is interesting to observe from \cref{fig:exp} and Table~\ref{tab:lt-scale} that the speedups of the PIP algorithms over the PBBS implementations on one thread is larger than on 72 cores with hyper-threading.
We conjecture that on one thread, the simpler logic improves prefetching, whereas when using all hyper-threads, prefetching does not help as much due to the memory bandwidth already being saturated.

\subsection{Additional Space Usage}

 Table~\ref{tab:memory-usage} shows the input size and  total memory usage of our PIP algorithms.
We see that for our two \strong{} algorithms (scan and filter), the auxiliary space overhead is negligible (less than 0.1\%), and for the three \weak{} algorithms  (random permutation, list contraction, and tree contraction), the best performance is achieved when the space overhead is between 0.9--3.7\%, which is still much smaller than the input size.
We can further reduce the space overhead for the \weak{} algorithms at the cost of higher running time (e.g., see Table~\ref{tab:rp-additional-space}).
In contrast, the existing non-in-place algorithms for these problems require  additional space proportional to the input size.

\begin{table}[t]
  %\small
  \centering
  \begin{tabular}{@{  }>{\centering}p{3cm}@{  }>{\centering}p{1.35cm}>{\centering}p{1.6cm}p{1cm}<{\centering}}
    \toprule
    Problem & Input size (MB) & Memory usage (MB) & Over-head \\
    \midrule
    Scan  & 7629.4 & 7636.2 & $<$0.1\% \\
    Filter & 7629.4 & 7636.9 & $<$0.1\% \\
    Random permutation & 762.9 & 791.2 & 3.7\% \\
    List contraction & 762.9 & 773.5 & 1.4\% \\
    Tree contraction & 1144.4 & 1154.9 & 0.9\% \\
    \bottomrule
    \end{tabular}%
  \caption{Memory usage of our algorithms on the five experiments  from \cref{fig:summary}.}
  \label{tab:memory-usage}%
  \end{table}

\section{Conclusion}
In this paper, we defined two models for analyzing parallel
in-place algorithms. We presented new parallel in-place
algorithms for scan, filter, partition, merge, random permutation,
list contraction, tree contraction, connectivity, biconnectivity, and
minimum spanning forest. We  implemented several of our algorithms, and
showed experimentally that they are competitive or outperform
state-of-the-art non-in-place parallel algorithms for the same
problems.

\subsection*{Acknowledgements}
We thank Guy Blelloch for letting us use his group's machine for our experiments.
This research was supported by DOE Early Career Award \#DE-SC0018947, NSF
CAREER Award \#CCF-1845763, Google Faculty Research Award, DARPA SDH Award \#HR0011-18-3-0007, and
Applications Driving Architectures (ADA) Research Center, a JUMP
Center co-sponsored by SRC and DARPA.

\bibliographystyle{abbrv}
%\bibliography{ref}

\appendix
\section{Appendix}
\subsection{Deterministic Reservations}
We now describe the framework of deterministic reservations in more detail.
Each
round of deterministic reservations consists of a reserve phase, followed by a synchronization point, and then a commit phase.
In Algorithms~\ref{alg:parallel-rp} and~\ref{algo:lc},
the reserve phase is first parallel for-loop that
 writes to locations in a
shared data structure $R$, corresponding to the steps (iterates). This is used to resolve conflicts among different steps that may modify the same memory locations.
The commit phase is the second parallel for-loop that checks $R$ to see if the reservations for the step were successful (all of its writes end up as the final value in $R$); if so, the step is executed.
Iterates that fail to execute will be packed by the framework and will retry in the next round.
This is repeated until no iterates remain.

To achieve the best practical performance, instead of trying all iterates simultaneously (many of which will fail, and waste work), the framework only works on a prefix of all the iterates.
After each round, the failed iterates are packed and new iterates are added to the prefix so that we have a sufficient number of iterates for the next round.
In practice, we pick a prefix of 1--2\% of the overall number of iterates, which
gives a good trade-off between  work and parallelism, and at the same time it naturally meets our requirement for controlling the execution size in the \weak{} algorithm.

For PIP algorithms, we usually need an additional phase in each round, which we refer to as the \emph{cleaning phase}.
In classic parallel algorithms based on deterministic reservations, we only initialize the reservation array $R$ at the beginning of the algorithm (line 1 in Algorithm~\ref{alg:parallel-rp} and~\ref{algo:lc}). However,
for PIP algorithms, we need to clean the data and reuse the space for the next round.

\subsection{Random Permutation Implementations}
\myparagraph{\naive{}}.
Our first version uses parallel hash tables to replace the auxiliary arrays $R$ and $H$, and we refer to this implementation as this algorithm \naive{}.
Unfortunately, compared to \algname{RP-PBBS}, \naive{} has poor performance in practice.
On input sizes between 10 million to 1 billion, \naive{} requires 2.9--3.5x running time of the \algname{RP-PBBS} algorithm, as shown in Figure~\ref{fig:rp-exp}.
The reason is that, the original \algname{RP-PBBS} implementation from PBBS needs roughly 4 sequential accesses and 2 random accesses per swap on each round, and the 2 random accesses per swap is to set and check $R[H[i]]$ per round.
Since we cannot keep $R$ and $H$ explicitly but we need to use parallel hash tables, \naive{} incurs 6 random accesses per swap (setting and checking $R[i]$, $H[i]$, and $R[H[i]]$), and 3 more sequential accesses for cleaning. The cleaning for both $R$ and $H$ is done by simply clearing the entire hash tables, which we found to be more efficient than doing individual hash table deletions.
%This approximately slow down the performance by about a factor or three.

\myparagraph{\rpflat{}: packing $H[i]$ as an array}.
In \algname{RP-PBBS}, the value of $H[i]$ is computed by a hash function.
Since generating a good hash function in practice is expensive and this value is used in a variety of places, we store it in an array $H$ to avoid recomputation.
We note that in Algorithm~\ref{alg:parallel-rp}, the access of $H[i]$ is always associated with index $i$, and we  only need to access the values of $i$ that are the sources of the swaps in the prefix for the current round.
Hence, we only keep an array of the size of the prefix, and use it to store the corresponding value in $H$ for each iterate in the prefix.
As such, we need to modify the code of deterministic reservations so that it also provides the index of each swap $i$ in the overall list of active iterates, so that we can look up its value in $H$.
By doing so, we reduce the number of hash table inserts/queries of each swap from three to two (by eliminating the random accesses to $H[i]$).
The cleaning phase incurs  two serial accesses per swap since we need to clear the hash table for $H$, which is twice the size of the prefix (as done in \naive{}, we clear the entire hash table).
We refer to this implementation as \rpflat{}.

\myparagraph{\oneres{}: using only one reservation}.
The previous implementations use two updates per swap (Lines~\ref{line:line6} and~\ref{line:resend} in Algorithm~\ref{alg:parallel-rp}).
However, we observe that we can apply just one  update on Line~\ref{line:resend}, and avoid the update on Line~\ref{line:line6}.
Then, the if-condition on Line~\ref{line:ressuc} is modified to
$$(R[i] = \bot ~\mbox{||}~ R[i] = i) ~\mbox{\&\&}~ R[H[i]] = i$$
Here $\bot$ indicates the initialized value of hash table, meaning that key $i$ is not found in the hash table.
The advantage is that we reduce the number of hash table insertions of each swap from two to one, which enables us to use a hash table of half the size and reduce the overall memory footprint.
The cleaning phase incurs one sequential access per swap to clear the hash table.
We refer to this implementation as \oneres{}.

\myparagraph{\final{}: packing the consecutive part of $R[i]$ in an array}.
With the modified if-condition on Line~\ref{line:ressuc}, for the swap of iterate $i$ we need to check the value of $R[i]$.
We can use a similar approach as done for $H[i]$ in \rpflat{}, so that when accessing $R[i]$, instead of a random access to a hash table, we only need a sequential access to an array.
Unfortunately,
we do not have the inverse mapping from each swap to its index in the active list.
However, we do know that the indices are consecutive for all newly added swaps in a round.
Hence, we modify the code for deterministic reservations so that it provides the range of the consecutive indices and the range of the keys.
We can then use an array with the size of the active list.
For reserving $R[H[i]]$, we first check if the key falls in the range. If so, we map it to the associated index and update it in the array;
otherwise, we insert it into the hash table.
 Checking the values in the commit phase is done similarly. In this way, the accesses of $R[i]$ are mostly serial, and so we save one random access per swap on each round.
 For the cleaning phase, we clear the consecutive array, and for any swaps whose target is in the hash table, we delete it from the hash table (unlike our earlier implementations, we do not scan and clear the entire hash table).
Most of the swaps will incur a sequential access for clearing it in the consecutive array.
However, there is a random access to the hash table for clearing, if the swap target is in the hash table (i.e., it is not a newly added swap for this round).
This happens infrequently since most of the swaps are successful if we pick a small prefix in each round.
 We refer to this implementation as \final{}.
The approximate numbers of sequential and random memory accesses per swap on each round are shown in Table~\ref{tab:rp-count}.

\end{document}